\documentclass[12pt,cleveref]{colt2025} 
\usepackage{times}
\usepackage{graphicx} 
\usepackage[utf8]{inputenc}
\usepackage{enumerate}
\usepackage{dsfont}
\usepackage{nicefrac}
\usepackage{mathtools}
\usepackage{tikz}
\usepackage{xcolor}
\usepackage{paralist}
\usepackage{booktabs}
\usepackage{subcaption}
\usepackage{paralist}
\usepackage{wrapfig}
\usepackage{etoolbox}
\usepackage{xspace}
\usetikzlibrary{positioning, fit, shapes.geometric}

\newtheorem{observation}[theorem]{Observation}

\newcommand{\problemdef}[3]{%
  \begin{center}
    \begin{minipage}{0.95\columnwidth}
      \noindent
      \textsc{#1}
      \vspace{5pt}\\
      \setlength{\tabcolsep}{3pt}%
      \begin{tabular}{@{}l p{0.8\columnwidth}@{}}
        \textbf{Input:}    & #2 \\
        \textbf{Question:} & #3
      \end{tabular}
    \end{minipage}
  \end{center}
}

\newcommand{\XP}{\textrm{XP}\xspace}

\newcommand{\FPT}{\textrm{FPT}\xspace}
\newcommand{\NP}{\textrm{NP}\xspace}
\newcommand{\coNP}{\textrm{coNP}\xspace}

\newcommand{\N}{\mathbb{N}}
\newcommand{\Z}{\mathbb{Z}}
\newcommand{\Q}{\mathbb{Q}}
\newcommand{\R}{\mathbb{R}}

\newcommand{\conv}{\operatorname{conv}}
\newcommand{\pos}{\operatorname{pos}}
\newcommand{\spn}{\operatorname{span}}

\newcommand{\rank}{\operatorname{rank}}
\newcommand{\poly}{\operatorname{poly}}
\newcommand{\inner}[1]{\langle #1 \rangle}
\newcommand{\ReLUInj}{\textsc{ReLU-Layer Injectivity}\xspace}
\newcommand{\lReLUInj}{\textsc{$\ell$-ReLU-layer Injectivity}\xspace}
\newcommand{\ReLUPos}{\textsc{2-Layer ReLU Positivity}\xspace}
\newcommand{\ReLUNonInj}{\textsc{ReLU-Layer Non-Injectivity}\xspace}
\newcommand{\AcycDisc}{\textsc{Acyclic $2$-Disconnection}\xspace}

\Crefname{theorem}{Theorem}{Theorems}
\crefname{theorem}{Thm.}{Thms.}
\Crefname{corollary}{Corollary}{Corollaries}
\crefname{corollary}{Cor.}{Cors.}
\Crefname{algorithm}{Algorithm}{Algorithms}
\Crefname{figure}{Figure}{Figures}
\Crefname{observation}{Observation}{Observations}

\newenvironment{proofwithcaption}[1]{
\renewcommand{\proofname}{#1}
\begin{proof}
}{
\end{proof}
}
\title[Injectivity and Verification of ReLU Neural Networks]{Complexity of Injectivity and Verification of ReLU Neural Networks}

\coltauthor{%
 \Name{Vincent Froese} \Email{vincent.froese@tu-berlin.de}\\
 \addr Technische Universit\"at Berlin, Faculty~IV\\
Institute of Software Engineering and Theoretical Computer Science, Algorithmics and Computational Complexity.\\
\protect\\
 \AND
 \Name{Moritz Grillo} \Email{grillo@math.tu-berlin.de}\\
 \Name{Martin Skutella} \Email{martin.skutella@tu-berlin.de}\\
 \addr Technische Universit\"at Berlin, Faculty~II, Institute of Mathematics,\protect\\ Combinatorial Optimization and Graph Algorithms.\protect\\
}

\begin{document}

\maketitle

   \begin{abstract} 
   Neural networks with ReLU activation play a key role in modern machine learning. Understanding the functions represented by ReLU networks is a major topic in current research as this enables a better interpretability of learning processes. 

    Injectivity of a function computed by a ReLU network, that is, the question if different inputs to the network always lead to different outputs, plays a crucial role whenever invertibility of the function is required, such as, e.g., for inverse problems or generative models. The exact computational complexity of deciding injectivity was recently posed as an open problem (Puthawala et al.~[JMLR~2022]). We answer this question by proving \coNP-completeness. On the positive side, we show that the problem for a single ReLU-layer is still tractable for small input dimension; more precisely,
    we present a parameterized algorithm which yields fixed-parameter tractability with respect to the input dimension. 
    
In addition, we study the network verification problem which is to verify that certain inputs only yield specific outputs. This is of great importance since neural networks are increasingly used in safety-critical systems. We prove that network verification is coNP-hard for a general class of input domains. Our results also exclude constant-factor polynomial-time approximations for
         the maximum  of a function computed by a ReLU network.
         
In this context, we also characterize surjectivity of functions computed by ReLU networks with one-dimensional output which turns out to be the complement of a basic network verification task. 
          We reveal interesting connections to computational convexity by formulating the surjectivity problem as a zonotope containment problem.\footnote{Accepted for presentation at the
Conference on Learning Theory (COLT) 2025}
	\end{abstract}
	
	\section{Introduction}
         Neural networks with rectified linear units (ReLUs) are a widely used model in deep learning. In practice, neural networks are trained on finite datasets and are expected to generalize to new, unseen inputs. However, they often exhibit unexpected and erroneous behavior in response to minor input perturbations; see, e.g., \cite{SzegedyZSBEGF13}.
        Hence, the certification of trained networks is of great importance and necessitates a thorough understanding of essential properties of the function computed by a ReLU network.
        
        Network verification, that is, the question whether for all inputs from a given subset~$\mathcal X$, the ReLU network outputs a value contained in a given set~$\mathcal Y$, is a research field gaining high interest recently since neural networks are increasingly used in safety-critical systems like autonomous vehicles (\cite{bojarski2016endendlearningselfdriving}) and collision avoidance for drones (\cite{aircraft}), see, e.g.,~\cite{Weng18,RP21,KL21,KBDJK22}. 
        Typically,~$\mathcal X$ and~$\mathcal Y$ are balls or defined by some linear constraints and the problem is known to be \coNP-hard if the sets $\mathcal{X}$ and $\mathcal{Y}$ are part of the input; see~\cite{KBDJK22,Weng18,salzer2023}.

        Moreover, recent works focus on studying elementary properties of functions computed by ReLU networks such as injectivity; see, e.g.,~\cite{PKLDH22,HEB23,FPLH23}.
	A function is called \emph{injective} if no two different inputs are mapped to the same output. Injectivity plays a crucial role in many applications where invertibility of a neural network is necessary (\citet{Gomez,Kingma,ardizzone2018analyzing,InverttoInvert}). Such applications include, for example, generative models, inverse problems, or likelihood estimation; see \cite{PKLDH22,FPLH23} for a more detailed discussion. On the other hand, injectivity might cause privacy issues since the input might be inferred from the output.
   
        Formally, for any number of layers $\ell \in \N$, given weights of a fully connected ReLU network $f\colon \R^d \to \R^m$ with $\ell$ layers, the question is whether the map~$f$ is injective.
        \citet{PKLDH22} give a mathematical characterization for injectivity of a single ReLU-layer (without bias), which implies an exponential-time algorithm where the exponential part is upper bounded by~$\min\{2^m,m^d\}$.
        In terms of \emph{parameterized complexity}, the problem is \emph{fixed-parameter tractable}
        with respect to the number~$m$ of ReLUs; that is, the superpolynomial part depends only on~$m$.
        With respect to the input dimension~$d$, the problem is in the complexity class XP; that is, it can be solved in polynomial time for any constant~$d$.
        The complexity status of deciding injectivity of an $\ell$-layer ReLU neural network (and therefore also, in particular, a  single layer), however, is unresolved and posed as an open problem by \cite{PKLDH22}.

        Another natural and fundamental property of functions to consider is surjectivity. A function \mbox{$f:\R^d\to\R^m$} is called \emph{surjective} if~$f(\R^d)=\R^m$. 
        Hence, to complete the picture alongside injectivity, we initiate the study of surjectivity for ReLU networks, which, to the best of our knowledge, has not been considered before.  In fact, we show that deciding surjectivity can be considered as the complement of a certain network verification task.
        \vspace{-0.25cm}
        \paragraph*{Our Contributions.}
        \label{paragraph:contributions}
        After some preliminaries in \Cref{sec:prelims},
        we show in \Cref{sec:complexity-injectivity} that deciding injectivity of a ReLU network with $\ell$ layers is \coNP-complete (\Cref{cor:injcoNPcomplete}) for any $\ell \in \N$ and
        thus not polynomial-time solvable, unless P $=$ NP. Notably, our hardness reduction reveals some interesting connections between cut problems in (di)graphs and properties of ReLU networks via graphical hyperplane arrangements. We believe that these connections between seemingly unrelated areas are of independent interest.
        Moreover, our hardness result implies a running time lower bound of~$2^{\Omega(m)}$ for a single ReLU-layer with $m$ neurons based on the Exponential Time Hypothesis (\Cref{cor:ReLU-lower}).
        Hence, the running time~$2^m$ is essentially optimal.
        As regards the input dimension~$d$, however, in \Cref{sec:FPT-algorithm} we give an improved algorithm running in~$O\bigl((d+1)^d\cdot\poly\bigr)$ time. This yields fixed-parameter tractability (\Cref{thm:fpt})
        and settles the complexity status for injectivity of a ReLU-layer.
    
         \citet{KBDJK22} and \citet{salzer2023} show that the network verification problem is \coNP-hard for a single hidden layer. However, their hardness results are based on the unit cube as input set. Thus, it was not clear so far whether there are other special input sets on which network verification is solvable in polynomial time. In \Cref{sec:surjectivity}, we provide the strongest hardness result for network verification to date (\Cref{prop:verification_hard}), proving coNP-hardness for every possible input domain $\mathcal{X}$ that contains a ball  (of possibly lower dimension under some mild conditions). In particular, our result implies that the computational intractability does not stem from choosing a particular set~$\mathcal{X}$, but is rather an intrinsic property of a single hidden ReLU-layer. Since arguably any reasonable input domain contains some ball (this holds, e.g., for all polyhedral sets satisfying the same mild conditions), no special case of verification regarding the input domain is solvable in polynomial time in the worst case. In particular, our result implies hardness for all polyhedra and also all balls as input set, which is not covered by the previous hardness results. In addition, our results also exclude constant-factor polynomial-time approximations for
         the maximum. As a consequence, one must make (very) specific assumptions on the network in order to obtain tractable cases, e.\;g., by building special (approximating) networks that are efficiently verifiable (\citet{Baader2020Universal},\cite{Zi},\cite{baader2024expressivity}).
 
         Moreover, we give a characterization of surjectivity for ReLU networks with one-dimensional output~(\Cref{lemma:charsurjectivity}) which implies a polynomial-time algorithm for constant input dimension~$d$.
        We then proceed with proving \NP-hardness of surjectivity by showing that it can be phrased as (the complement of) a special case of network verification.
        Finally, we also show that surjectivity can be formulated as a zonotope containment problem, which is of fundamental importance, e.g., in robotics (\cite{KULMBURG202184}).
 
        \vspace{-0.25cm}
	\paragraph*{Related Work.}
        \citet{PKLDH22} initiate the study of injectivity of ReLU networks. Their result implicitly yields an exponential-time algorithm.
	\citet{HEB23} use a frame-theoretic approach to study the injectivity of a ReLU-layer on a closed ball.
        They give an exponential-time algorithm to determine a bias vector for a given weight matrix such that the corresponding map is injective on a closed ball.
	\citet{FPLH23} study injectivity of ReLU-layers with linear neural operators.

       The network verification problem is known to be \coNP-hard for networks with two layers and arbitrary output dimension (\citet{KBDJK22} proved \NP-hardness for the complement). \citet{salzer2023} corrected some flaws in the proof and showed \coNP-hardness for one-dimensional output.
        The problem is also known to be \coNP-hard to approximate in polynomial time for one-dimensional output if the number of layers is unbounded~(\cite{Weng18}) and for three hidden layers and squashable activation functions~\cite{Zi}. Efficient verifiable networks with heuristics such as the 
 interval bound propagation \cite{Gowal2018OnTE} are shown to be universal approximators~\cite{Zi,Baader2020Universal}, but the networks might have exponential size. On the other hand, obtaining exact representations for all functions computed by ReLU neural networks is not possible such that they are precisely verifiable with interval bound propagation~\cite{mirman2022the} or single neuron convex relaxations~\cite{baader2024expressivity}.
Further heuristics include methods such as DeepPoly (\cite{Singh}), DeepZ (\cite{Wong}), general cutting planes (\cite{Cuttingplane}), multi neuron verification (\cite{ferrari2022complete}). Moreover, there exist libraries for verification (\cite{Library1,mao2024understandingcertifiedtraininginterval}).

        Parameterized complexity has also been studied for the training problem for ReLU neural networks by~\cite{FHN22,FH23}.
        In general, understanding the complexity/expressivity of ReLU neural networks is an important task~(\cite{HBDS21,AroraBMM18}).
	
	\section{Preliminaries}\label{sec:prelims}

        We introduce the basic definitions and concepts involving ReLU neural networks and their geometry which are relevant for this work.
	
	\begin{definition}
		A ReLU-layer with $d$ inputs, $m$ outputs, weights~$\mathbf{W}\in\R^{m\times d}$, and biases~$\mathbf{b}\in\R^m$ computes a map
			$\phi_{\mathbf{W},\mathbf{b}}\colon \R^d \to \R^m,\quad \mathbf{x}\mapsto[\mathbf{W}\mathbf{x}+\mathbf{b}]_+$,
		where $[\cdot]_+ \colon \R^m \to \R^m$ is the \emph{rectifier function} given by $[\mathbf{x}]_+\coloneqq(\max\{0,x_1\},\ldots,\max\{0, x_m\})$.
	\end{definition}
    In case that $\mathbf{b}=\mathbf{0}$, we omit~$\mathbf{b}$ and simplify notation~$\phi_{\mathbf{W}}\coloneqq\phi_{\mathbf{W},\mathbf{0}}$.  A deep ReLU network is just a concatenation of compatible ReLU-layers.
 \begin{definition}
     An $\ell$-layer ReLU network of architecture $(d=n_0,n_1,\ldots,n_{\ell-1},n_{\ell}=m)$ with weights $\mathbf{W}_{i} \in \R^{n_{i-1} \times n_{i}}$ and biases $\mathbf{b}_{i} \in \R^{n_{i}}$ for $i \in [\ell]$ computes a map
     \begin{align*}
			f \colon& \R^{d} \to \R^{m},\quad \mathbf{x}\mapsto \mathbf{W}_{\ell}\cdot(\phi_{\mathbf{W}_{\ell-1},\mathbf{b}_{\ell-1}} \circ \cdots \circ \phi_{\mathbf{W}_{1},\mathbf{b}_{1}})(\mathbf{x})+\mathbf{b}_{\ell}.
        \end{align*}
 \end{definition}
 Whenever we consider an $\ell$-layer ReLU neural network $f \colon \R^d \to \R^m$ as the input to a decision problem in the paper, we implicitly consider its weights and biases as the input.
%
\begin{wrapfigure}{r}{2.5cm}
   \begin{tikzpicture}[scale=0.6]
\draw[->] (-2.5,0) -- (2.5,0);
\draw[->] (0,-2.5) -- (0,2.5);

\draw[thick] (2,2) -- (-2,-2) node[anchor=north] {$H_{\mathbf{w}_1}$};
\draw[->] (-1,-1) -- (-1.2,-0.8);
\fill[red,opacity=0.1] (-2,-2) -- (2,2) -- (2,2) -- (-2,2) -- cycle;

\draw[thick] (-1,-2) -- (1,2) node[anchor=south] {$H_{\mathbf{w}_2}$};
\draw[->] (0.5,1) -- (0.7,0.9);
\fill[blue,opacity=0.1] (-1,-2) -- (1,2) -- (2,2) -- (2,-2) -- cycle;

\draw[thick] (1.5,-2) -- (-1.5,2) node[anchor=south east] {$H_{\mathbf{w}_3}$};
\draw[->] (0.75,-1) -- (1.05,-0.8);
\fill[yellow,opacity=0.1] (1.5,-2) -- (-1.5,2) -- (2,2) -- (2,-2) -- cycle;
\end{tikzpicture}
\caption{The polyhedral fan induced by an oriented linear hyperplane arrangement given by matrix $\mathbf{W} \in \R^{3 \times 2}$.}
 \label{fig:pol_fan}
\vspace{-0.5cm}
\end{wrapfigure}

\paragraph*{Geometry of ReLU-layers.}
    We review basic definitions from polyhedral geometry; see \citet{ilp_theory} for more details.
    For a vector $(\mathbf{w}_i,b_i) \in \R^{d+1}$, the hyperplane $H_{\mathbf{w}_i,b_i} \coloneqq \{\mathbf{x} \in \R^d \mid \inner{\mathbf{w}_i,\mathbf{x}}+b_i=0\}$ subdivides $\R^d$ into half-spaces $H^+_{\mathbf{w}_i,b_i} \coloneqq \{\mathbf{x} \in \R^d \mid \inner{\mathbf{w}_i,\mathbf{x}}+b_i\geq 0\}$ and $H^-_{\mathbf{w}_i,b_i} \coloneqq \{\mathbf{x} \in \R^d \mid \inner{\mathbf{w}_i,\mathbf{x}}+b_i\leq 0\}$.
    A \emph{polyhedron} $P$ is the intersection of finitely many closed half-spaces. A \emph{polyhedral cone} $C \subseteq \R^d$ is a polyhedron such that $\lambda u + \mu v \in C$ for every $u,v \in C$ and  $\lambda, \mu \in \R_{\geq 0}$.
     A hyperplane \emph{supports} $P$ if it bounds a closed half-space containing $P$, and
any intersection of $P$ with such a supporting hyperplane yields a \emph{face} $F$ of $P$.
    A \emph{polyhedral complex} $\mathcal P$ is a finite collection of polyhedra such that (i) $\emptyset \in \mathcal P$, (ii) if $P \in \mathcal P$ then all faces of $P$ are in $\mathcal P$, and (iii) if $P, P' \in \mathcal P$, then $P \cap P'$ is a face of both $P$ and $P'$. A \emph{polyhedral fan} is a polyhedral complex in which every polyhedron is a cone.  
    For a matrix $\mathbf{W} \in \R^{m \times d}$ and $\mathbf{b} \in \R^m$, the set of full-dimensional polyhedra (called \emph{cells})
    $\mathcal{C}_{\mathbf{W},\mathbf{b}} \coloneqq \biggl\{\bigcap_{i=1}^{m} H^{s_i}_{\mathbf{w}_i,b_i}\,\bigg|\, (s_1, \ldots, s_m) \in \{+,-\}^m,\, \dim\Bigl(\bigcap_{i=1}^{m} H^{s_i}_{\mathbf{w}_i,b_i}\Bigr)=d\biggr\}$
    subdivides $\R^d$ and the set $\Sigma_{\mathbf{W},\mathbf b}$ containing all cells and all their faces forms a polyhedral complex.
    If $\mathbf{b}=\mathbf{0}$, then $\Sigma_{\mathbf{W},\mathbf{b}}$ is a polyhedral fan. See  \Cref{fig:pol_fan} for an illustration of a $2$-dimensional polyhedral fan arising from a ReLU-layer.

    For a cell~$C\in\mathcal C_{\mathbf{W},\mathbf{b}}$, we define the \emph{active} set \mbox{$I_C\coloneqq\{j\in [m]\mid \forall\mathbf x \in C:\mathbf w_j\mathbf x +b_j \ge 0\}$ and} $\mathbf W_C\in\R^{m\times d}$, $\mathbf{b}_C \in \R^m$, with
		$(\mathbf{W}_C)_j \coloneqq  \begin{cases}
		\mathbf{w}_j, & j\in I_C,\\
		\mathbf{0}, & j\in[m]\setminus I_C,
              \end{cases}
        $ 
        $(\mathbf{b}_C)_j \coloneqq  \begin{cases}
		b_j, & j\in I_C,\\
		0, & j\in[m]\setminus I_C.
              \end{cases}$
        Note that the map $\phi_{\mathbf{W},\mathbf{b}}$ is affine linear on~$C$, namely $\phi_{\mathbf{W},\mathbf{b}}(\mathbf{x}) =\mathbf{W}_C\cdot\mathbf{x}+\mathbf{b}_C$ for~$\mathbf{x}\in C$.

        \paragraph*{Geometry of ReLU Neural Networks.}
        It is well-known that also a deep ReLU neural network $f$ partitions its input space into polyhedra on which $f$ is affine linear (\cite{HaninR19,grigsby}). More precisely, let $f_{i,j} = \pi_j \circ \phi_{\mathbf{W}_{i},\mathbf{b}_{i}} \circ \cdots \circ \phi_{\mathbf{W}_{1},\mathbf{b}_{1}}$, where $\pi_j \colon \R^{n_i} \to \R$ is the projection onto the $j$-th coordinate. Then, every $\mathbf{s}=(\mathbf{s}_1,\ldots,\mathbf{s}_{n_{\ell-1}}) \in\{-,+\}^{n_1} \times \ldots\times \{-,+\}^{n_{\ell-1}}$ corresponds to a (possibly empty) polyhedron \[P_{\mathbf{s}} =  \bigcap_{(\mathbf{s}_{i})_j = +} \{\mathbf x \in \R^d \mid f_{i,j}(\mathbf x) \geq 0\} \cap \bigcap_{(\mathbf{s}_{i})_j = -} \{\mathbf x \in \R^d \mid f_{i,j}(\mathbf x) \leq 0\}.\]
        The maps $f_{i,j}$ are affine linear on $P_{\mathbf{s}}$ and the coefficients are polynomially bounded in the weights and biases of the neural network.
        Hence, the polyhedron $P_{\mathbf{s}}$ arises as the intersection of at most $\sum_{i=1}^\ell n_\ell$ closed half-spaces whose encoding sizes are polynomially bounded in the weights and biases of~$f$. The set of all polyhedra $P_\mathbf{s}$ together with their faces forms a polyhedral complex which we denote by $\Sigma_f$.

        A \emph{ray} $\rho$ is a one-dimensional pointed cone; a vector $\mathbf r$ is a \emph{ray generator} of $\rho$ if $\rho = \{\lambda\, \mathbf r \mid \lambda \geq 0\}$.
         For a map $f$ computed by a ReLU network without bias and each ray~$\rho$ in~$\Sigma_f$, let $\mathbf r_\rho$ be the unique unit ray generator of~$\rho$ having norm $1$ and let $\mathcal{R} \coloneqq \{\mathbf r_\rho \mid \rho \text{ is a ray in } \Sigma_f\}$. A ray $\rho \subseteq C$ of a cone~$C$ is an \emph{extreme ray} if there do not exist $\lambda_1,\lambda_2 >0$ and $\rho_1,\rho_2 \subseteq C$ such that $\rho = \lambda_1 \rho_1 + \lambda_2 \rho_2.$
        We make the following simple observation.
        
	\begin{observation}
          \label[observation]{observation:linearoncells}        
          Let $C$ be a pointed polyhedral cone such that $f$ is linear on $C$ and let $\mathbf{r}_1, \ldots, \mathbf{r}_\ell$ be ray generators of the extreme rays of~$C$.
          Then, for each~$\mathbf x\in C$, there are $\lambda_1, \ldots, \lambda_\ell\in\R$ such that $\mathbf{x} = \sum_{i=1}^{\ell} \lambda_i \mathbf{r}_i$ and $f(\mathbf{x}) =\sum_{i=1}^{\ell}\lambda_i f (\mathbf{r}_i)$. 
        \end{observation}
        \noindent
\Cref{observation:linearoncells} implies that $f$ is determined by its values on $\mathcal{R}$, if all cones $C \in \Sigma_f$ are pointed.

    \paragraph*{(Parameterized) Complexity Theory.}
    We assume the reader to be familiar with basic concepts from classical complexity theory like P, \NP, and \NP-completeness. The class \coNP contains all decision problems whose complement is in \NP. A decision problem is \coNP-complete if and only if its complement is \NP-complete.
    Clearly, a \coNP-complete problem cannot be solved in polynomial time unless P $=$ \NP.

    A \emph{parameterized} problem consists of instances~$(x,k)$, where~$x$ encodes the classical instance and~$k\in\N$ is a \emph{parameter}. A parameterized problem is in the class \XP if it is polynomial-time solvable for every constant parameter value, that is, in $O(|x|^{f(k)})$ time for an arbitrary function~$f$ depending only on the parameter~$k$. A parameterized problem is \emph{fixed-parameter tractable} (contained in the class \FPT) if it is solvable in~$f(k)\cdot |x|^{O(1)}$ time, for an arbitrary function~$f$.
        Clearly, \FPT $\subseteq$ \XP; see \cite{DF13} for further details of parameterized complexity.

	\section{coNP-Completeness of Injectivity}
    \label{sec:complexity-injectivity}
        In this section we study the computational complexity of \lReLUInj, which is the problem to decide whether an $\ell$-layer ReLU network $f\colon \R^d\to\R^m$ computes an injective map.

        \begin{proposition}
            For every $\ell \in \N$, it holds that \lReLUInj is contained in \coNP.
        \end{proposition}
        \begin{proofwithcaption}{Proof sketch.}
        Two polyhedra $P,Q \in \Sigma_f$ serve as a certificate for non-injectivity, since checking whether there are $\mathbf{x} \in P$ and $\mathbf{x}' \in Q$ with $\mathbf{x} \neq \mathbf{x}'$ such that $f(\mathbf{x}) = f(\mathbf{x}')$ is simply checking feasibility of a linear program. For a rigorous proof, one can apply Theorem~3.3 by \cite{salzer2023}.
        \end{proofwithcaption}
	To show that \lReLUInj is \coNP-hard for any $\ell \in \N$, it suffices to show that deciding injectivity is already \coNP-hard for a single ReLU-layer. Hence, in the remainder of this section we study the decision problem \ReLUInj: Given a matrix $\mathbf{W} \in \R^{m\times d}$ and a vector $\mathbf{b} \in \R^m$, is the map $\phi_{\mathbf{W},\mathbf{b}}$ injective?
        
         \citet[Theorem~2]{PKLDH22} prove that a map $\phi_{\mathbf{W}}$ with~$\mathbf{W}\in\R^{m\times d}$ computed by a single ReLU-layer without bias (that is, $\mathbf{b}=\mathbf{0}$) is injective if and only if $\mathbf{W}_C$ has (full) rank~$d$ for all $C \in \mathcal{C}_{\mathbf{W}}$.
         Note that \cite[Lemma~3]{PKLDH22} state that one can assume~$\mathbf{b}=\mathbf{0}$ without loss of generality;
         more precisely, they claim that $\phi_{\mathbf{W},\mathbf{b}}$ is injective if and only if $\phi_{\mathbf{W}^+}$ is injective, where~$\mathbf W^+$ is obtained from~$\mathbf W$ by setting all rows with a negative bias to zero. 
         This argument, however, contains a logical inconsistency, since the map $x \mapsto [(x-1,-x+1)]_+$ is injective but $x\mapsto [(0,-x)]_+$ is not.
    Nevertheless, an analogous characterization can be shown for arbitrary bias.
	\begin{theorem}
		\label{theorem:characterization}
		A ReLU-layer $\phi_{\mathbf{W},\mathbf{b}}$ with~$\mathbf{W}\in\R^{m\times d}$ and $\mathbf{b}\in\R^m$ is injective if and only if $\mathbf{W}_C$ has rank~$d$ for all~$C\in\mathcal C_{\mathbf{W},\mathbf{b}}$.
	\end{theorem}
    \begin{proof}
    If there is a $C\in\mathcal C_{\mathbf{W},\mathbf b}$ such that $\mathbf{W}_C$ has rank $k < d$, then there exists an $\mathbf{x}^\perp\neq \mathbf 0$ with $\mathbf{W}_C\cdot\mathbf x^\perp=\mathbf 0$.  Now, let $\mathbf{x}$ be in the interior of $C$. Then, there is a $\lambda > 0$ such that $\mathbf{x}+\lambda \mathbf{x}^\perp \in C$. Hence, it follows that $\phi_{\mathbf{W},\mathbf{b}}(\mathbf{x}+\lambda \mathbf{x}^\perp) = \mathbf{W}_C\cdot(\mathbf{x}+\lambda \mathbf{x}^\perp)+\mathbf{b}_C = \mathbf{W}_C\cdot\mathbf{x}+\mathbf{b}_C=\phi_{\mathbf{W},\mathbf{b}}(\mathbf{x}).$

    Conversely, if there are $\mathbf{x},\mathbf{y} \in \R^d$ such that $\phi_{\mathbf{W},\mathbf{b}}(\mathbf{x}) = \phi_{\mathbf{W},\mathbf{b}}(\mathbf{y})$, then in particular $\inner{\mathbf{w},\mathbf{x}}+\mathbf{b}$ and $\mathbf{W}\mathbf{y}+\mathbf{b}$ have the same sign in each coordinate. Thus, $\mathbf x$ and $\mathbf y$ are both contained in a cell $C \in \mathcal C_{\mathbf{W},\mathbf{b}}$. Hence, it holds that $\mathbf{W}_C\cdot\mathbf{x}+\mathbf{b}_C = \mathbf{W}_C\cdot\mathbf{y}+\mathbf{b}_C$. Since $\mathbf{W}_C$ has rank~$d$, it follows that $\mathbf{x}=\mathbf{y}$.
    \end{proof}
	
    Since any hyperplane arrangement with~$m$ hyperplanes in~$d$ dimensions defines at most~$O(m^d)$ (also clearly at most~$2^m$) cells (\cite{Zas75}), \Cref{theorem:characterization} implies an algorithm that solves \ReLUInj in~$O(\min\{2^m,m^d\}\cdot\poly(S))$ time, where~$S$ denotes the input size.
 
    We prove that \ReLUInj is \coNP-complete.
	To this end, we show \NP-completeness for the complement problem \textsc{ReLU-layer Non-Injectivity}: Given~$\mathbf{W}\in\R^{m\times d}$ and~$\mathbf{b}\in\R^m$, is there a cell~$C\in\mathcal{C}_{\mathbf{W},\mathbf{b}}$ with $\rank(\mathbf{W}_C)< d$.

	For the \NP-hardness, we reduce from the following directed graph (\emph{digraph}) problem. A digraph $D=(V,A)$ is called \emph{acyclic} if it does not contain directed cycles and \emph{weakly connected} if the underlying graph, i.e., the graph where we have an undirected edge for every arc, is connected. The \AcycDisc problem is, given a digraph $D=(V,A)$, to decide whether there is a subset~$A'\subseteq A$ of arcs such that~$(V,A')$ is acyclic and~$(V,A\setminus A')$ is not weakly connected.
    The problem is a special case of \textsc{Acyclic $s$-Disconnection} where the goal is to remove an acyclic arc set such that the remaining digraph contains at least~$s$ weakly connected components. This more general problem is known to be \NP-hard if~$s$ is part of the input (\cite{FHO17}). In Appendix~\ref{Proof:Acyclic_Np_hard} we prove \NP-hardness for our special case~$s=2$ (\Cref{thm:acyc-NP}).
    We remark that our reduction implies that \AcycDisc cannot be solved in~$2^{o(|D|)}$ time unless the Exponential Time Hypothesis\footnote{The Exponential Time Hypothesis asserts that \textsc{3-SAT} cannot be solved in $2^{o(n)}$ time where~$n$ is the number of Boolean variables in the input formula (\cite{IP01}).} fails (\Cref{cor:AD-lower}). 

    \begin{figure}
    \centering
\subfigure[For the ordering $3 \prec 1 \prec 2$ (on the left) and the ordering $3 \prec 2 \prec 1$ (on the right), the arcs that respect the ordering are dotted and colored in orange.]{\begin{tikzpicture}
            \node (0) at (0,0) [circle,draw]{$3$};
            \node (1) at (0,1.5) [circle,draw]{$1$};
            \node (2) at (1.5,1.5) [circle,draw]{$2$};
            \draw (1) edge [teal,->,very thick,bend left=15] (0);
            \draw (1) edge [teal,->,very thick,bend left=15] (0);
            \draw (0) edge [->,very thick,orange, dotted] (2); 
            \draw (0) edge [->,very thick,orange, dotted, bend left=15] (1);
            \draw (1) edge [->,very thick,orange, dotted] (2);
            \draw (1) edge [teal,->,very thick,bend left=15] (0);
            \end{tikzpicture}
            \hspace{0.7cm}
            \begin{tikzpicture}
            \node (0) at (0,0) [circle,draw]{$3$};
            \node (1) at (0,1.5) [circle,draw]{$1$};
            \node (2) at (1.5,1.5) [circle,draw]{$2$};
            \draw (1) edge [teal,->,very thick] (2);
            \draw (1) edge [teal,->,very thick,bend left=15] (0);
            \draw (0) edge [->,very thick,orange, dotted] (2); 
            \draw (0) edge [->,very thick,orange, dotted, bend left=15] (1);
            \end{tikzpicture}}
    \hspace{3em}
        \subfigure[The resulting ReLU network represented by the corresponding oriented hyperplane arrangements. For the gray colored cones, the inactive neurons are colored in orange.]{
        \begin{tikzpicture}[scale=0.9]
        \draw[thick] (-1,0) -- (1,0);
        \draw[thick] (0,-1) -- (0,1);
        \draw[thick] (-1,-1) -- (1,1);
        \fill[gray,opacity=0.3] (0,0) -- (0,1) -- (1,1) --cycle;
        \node (012) at (0.5,1.3)   {$0 \leq x_1 \leq x_2$};
        \draw[orange,thick,->] (0.5,0) -- (0.5,-0.2);
        \draw[orange,thick,->] (0,0.5) -- (-0.2,0.5);
        \draw[orange,thick,->] (0.5,0.5) -- (0.7,0.3);
        \draw[teal,thick,->] (0,0.5) -- (0.2,0.5);
        \end{tikzpicture}
               \begin{tikzpicture}[scale=0.9]
        \draw[thick] (-1,0) -- (1,0);
        \draw[thick] (0,-1) -- (0,1);
        \draw[thick] (-1,-1) -- (1,1);
        \fill[gray,opacity=0.3] (0,0) -- (1,0) -- (1,1) --cycle;
        \node (021) at (2.1,0.5)  {$0 \leq x_2 \leq x_1$};
        \draw[orange,thick,->] (0.5,0) -- (0.5,-0.2);
        \draw[orange,thick,->] (0,0.5) -- (-0.2,0.5);
        \draw[teal,thick,->] (0,0.5) -- (0.2,0.5);
        \draw[teal,thick,->] (0.5,0.5) -- (0.7,0.3);
            \end{tikzpicture}}
\caption{An illustration of the reduction from \textsc{Acyclic $2$-Disconnection} to \ReLUNonInj for a digraph with $n=3$ nodes. An ordering $\pi$ of the nodes induces an acyclic subset $A_\pi$ of arcs corresponding to inactive neurons on the cone $C_\pi=\{x_{\pi(1)} \leq x_{\pi(2)} \leq x_{\pi(3)}\}$. The active neurons on $C_\pi$ have full rank $2$ if and only if removing the arcs $A_\pi$ results in a weakly connected digraph.}
    \label{fig:reduction_acyc_to_inj}
    \end{figure}
	
    \begin{theorem}\label{thm:ReLUInj-NPhard}
		\ReLUNonInj is \NP-complete even if every row of $\mathbf{W}$ contains at most two non-zero entries and~$\mathbf{b}=\mathbf{0}$.
	\end{theorem}
  \begin{proofwithcaption}{Proof sketch.}
  Containment in \NP is easy: The set~$I_C$ of a cell~$C$ with~$\rank(\mathbf{W}_C)<d$ serves as a certificate.
  To prove \NP-hardness, we reduce from \AcycDisc which is \NP-hard by \Cref{thm:acyc-NP}. We sketch the proof here, a detailed proof can be found in Appendix~\ref{proof:ReLUInj-NPhard}. Moreover, the reduction is illustrated in \Cref{fig:reduction_acyc_to_inj}.
  Given a digraph $D=(V,A)$ with $V=[n]$, we define a ReLU-layer $\phi \colon \R^{n-1} \to \R^{|A|}$ by $\phi(\mathbf{x}) \coloneqq ([x_i-x_j]_+)_{(i,j) \in A}$, where we set~$x_n\coloneqq 0$.
  
  For a permutation~$\pi \in \mathcal{S}_{n}$, let $A_\pi \coloneqq \{(i,j) \in A \mid \pi^{-1}(i) < \pi^{-1}(j)\}$ be the (acylic) subset of arcs respecting the (total) order on the nodes induced by $\pi$ and let $\mathbf x\in\R^{n-1}$ be a point with $x_{\pi(1)} < \dots < x_{\pi(n)}$.
  Note that the arcs in $A \setminus A_\pi$ correspond exactly to the neurons that are active on the cell~$C$ containing $\mathbf x$, i.e., neurons computing $[x_i-x_j]_+$ where $\pi^{-1}(j) < \pi^{-1}(i)$ and hence $x_j \leq x_i$.
  In Appendix~\ref{proof:ReLUInj-NPhard}, we prove that
$D_\pi \coloneqq (V,A \setminus A_\pi)$  is weakly connected, that is,  there is a path in the underlying graph of $D_\pi$ from node $j$ to node $n$ for all $j \in [n-1]$, if and only if $\mathbf{W}_{C}$ has full rank $n-1$.
  Hence, it holds that $\phi$ is not injective if and only if there exists $A' \subseteq A$ such that~$(V,A')$ is acyclic and~$(V,A\setminus A')$ is not weakly connected, proving the correctness of the reduction.    
  \end{proofwithcaption}


  The lower bound from \Cref{cor:AD-lower} actually transfers to \ReLUInj since our polynomial-time reduction in the proof of \Cref{thm:ReLUInj-NPhard} yields a matrix where the number of rows/columns is linear in the number arcs/nodes of the digraph.

	\begin{corollary}\label{cor:Inj-coNP}\label{cor:ReLU-lower}
  \ReLUInj is \coNP-complete even if every row of $\mathbf{W}$ contains at most two non-zero entries (and~$\mathbf{b}=\mathbf{0}$).
 Moreover, \ReLUNonInj and \ReLUInj cannot be solved in~$2^{o(m+d)}$ time, unless the ETH fails.
	\end{corollary}

To prove \coNP-hardness for an arbitrary number~$\ell$ of layers, one can simply reduce \ReLUInj to \lReLUInj by concatenating the layer with a ReLU network with $\ell-1$ layers that computes the identity map. Hence, we obtain the following corollary. 
 \begin{corollary}
 \label{cor:injcoNPcomplete}
     For every $\ell \geq 2$, it holds that \lReLUInj is \coNP-complete.
 \end{corollary}
	
	\section{An FPT Algorithm for ReLU-Layer Injectivity}
    \label{sec:FPT-algorithm}

        The hardness results from the previous section exclude polynomial running times as well as running times subexponential in~$m + d$ for deciding injectivity of a single ReLU-layer.
        Nevertheless, we show in this section that the previous upper bound of~$m^d$ can be improved to~$(d+1)^d$.
        We achieve this with a branching algorithm which searches for a ``non-injective'' cell~$C$, that is, $\mathbf W_C$ has rank strictly less than~$d$. 
        The algorithm branches on the ReLUs active in~$C$, thus restricting the search space to some polyhedron.
        The pseudocode is given in \Cref{Layerinjectivity,FindCell}.
        
	\begin{algorithm2e}[b]
          \caption{LayerInjectivity}
          \label{Layerinjectivity}
          \LinesNumbered
		\SetKwInOut{KwIn}{Input}
		\SetKwInOut{KwOut}{Output}
		
		\KwIn{$W =\{\mathbf w_1,\ldots,\mathbf w_m\}\subseteq \R^{d}$, $B=\{b_1,\ldots,b_m\} \subseteq \R$}
		\KwOut{$\mathbf x\in \R^d$ with~$\rank(\{\mathbf w_i\in W\mid\mathbf w_i\mathbf x+b_i\ge 0\})< d$ (if it exists); otherwise, ``yes''}
		
		\lIf{$\rank(W)< d$,}
		{\Return $\mathbf{0}$}\label{trivial1}
		
		\lIf{$\exists \mathbf x\in\R^d: \forall i\in[m]:\mathbf w_i\mathbf x +b_i < 0$,}
		{\Return $\mathbf{x}$}\label{trivial2}
			$x \gets$ FindCell($\emptyset$, $\emptyset$, $W$, $B$)\;\label{find}
			\lIf{$x\neq$ \textnormal{``no''},}{\Return $x$}
		\Return ``yes''
	\end{algorithm2e}
    
	\begin{algorithm2e}[t]
		\caption{FindCell}
		\label{FindCell}
        \LinesNumbered
		\SetKwFunction{isOddNumber}{isOddNumber}
		\SetKwInOut{KwIn}{Input}
		\SetKwInOut{KwOut}{Output}
		
		\KwIn{$C=\{\mathbf c_1,\ldots,\mathbf c_n\}\subseteq\R^{d}$, $P=\{p_1,\ldots,p_n\}\subseteq\R$ and $M=\{\mathbf m_1,\ldots,\mathbf m_m\} \subseteq \R^{d}, B=\{b_1,\ldots,b_m\} \subseteq \R$,}
		\KwOut{$\mathbf{x} \in \R^d$ (if it exists) such that $\mathbf c_i\mathbf{x} +p_i\ge 0$, for all~$i\in[n]$, and
			$\rank(C\cup \{\mathbf m_i\mid \mathbf m_i\mathbf x \ge 0\})< d$;
			otherwise, ``no''}
		\lIf{$\rank(C)=d$,}{\Return ``no''}\label{checkNo}
		
		$I \gets \bigl\{i \in \{1,\ldots,m\} \mid \textbf{m}_i \notin \spn(C)\bigr\}$\;
		$M \gets \{\mathbf m_i\mid i\in I\}, B \gets \{ b_i\mid i\in I\}$ \;\label{filter}
		
		\lIf{$\exists \mathbf x\in\R^d: \bigl((\forall i\in[n]:\mathbf c_i\mathbf x +p_i \ge 0) \wedge (\forall i\in I:\mathbf m_i\mathbf x +b_i< 0)\bigr)$,}
		{\Return $\mathbf{x}$}\label{returnX}
		
		{compute $A \subseteq I$ such that $\{\mathbf x \mid \forall i\in[n]:\mathbf c_i\mathbf x+p_i\ge 0\} \subseteq \bigcup_{i \in A} H^+_{\mathbf{m}_i,b_i}$ and $|A| \le d+1$}\;\label{cover}
		
		\ForEach{$i\in A$}{
			$x\gets$ FindCell$\left(C\cup\{\mathbf{m}_i\},P \cup \{b_i\}, M\setminus \{\mathbf m_i\}, B \setminus \{b_i\}\right)$\;\label{recurse}
			\lIf{$x\neq$ \textnormal{``no'',}}{\Return $x$}
		}
		\Return ``no''
	\end{algorithm2e}

        The key idea to bound the running time of our search tree algorithm is to show that there are always at most~$d+1$ candidate ReLUs to set active. The candidates are those whose corresponding half-spaces cover the current search space; that is, at least one of them will be active in the sought cell (cf.~\Cref{cover} of \Cref{FindCell}).
        Helly's Theorem (\cite{Helly1923}) ensures that there are always at most~$d+1$ such candidates. Using duality of linear programming, we show that these half-spaces can be found in polynomial time. A more detailed proof of the following lemma is given in Appendix~\ref{app:lem:conecover}.
	\begin{lemma}\label{lem:conecover}
		Let $P \subseteq \R^d$ be a polyhedron and let $\{(\mathbf{w}_1,b_1),\ldots ,(\mathbf{w}_n,b_n)\} \subseteq \R^{d+1}$ be such that $C$ is covered by the corresponding half-spaces, that is, $C \subseteq \bigcup_{i=1}^n H_{\mathbf{w}_i,b_i}^+$. Then, there exists a subset $A \subseteq [n]$ of size at most $d+1$ computable in polynomial time, such that $C \subseteq \bigcup_{i \in A} H_{\mathbf{w}_i,b_i}^+$.
	\end{lemma}	
    As regards the correctness of our algorithm,
    assume that there exists a vector~$\mathbf{x} \in \R^d$ such that $\rank(\{\mathbf{w}_i \mid \inner{\mathbf{w}_i,\mathbf{x}} +b_i \geq 0\}) = k < d$. Then, $\mathbf{x}$ is contained in some search space of FindCell, that is,~$\mathbf c_i\mathbf x + p_i \geq 0$ for all~$i\in[n]$. Notice that every branching increases the number of linear independent neurons that are active in the current search space by one. Thus, after at most $k$ branchings, FindCell finds a point of a cell where the map is not injective. Conversely, if the algorithm outputs an $\mathbf{x} \in \R^d$, then $\mathbf x$ is contained in a cell with rank less than~$d$; otherwise, \Cref{checkNo} in FindCell would have output ``no''. 
    
    \Cref{lem:conecover} ensures that we can always branch on $d+1$ neurons. Furthermore, the number of active neurons increases after every branching by one. Hence the search tree has at most $(d+1)^d$ nodes. We conclude our findings in the following theorem for which we provide a detailed proof in Appendix~\ref{app:thm:fpt}.
	\begin{theorem}\label{thm:fpt}
		\ReLUInj is solvable in $O((d+1)^d\cdot \poly(S))$ time, where~$S$ denotes the input size.
	\end{theorem}

 Deciding injectivity for a deep neural network $f= \phi_\ell \circ \cdots \circ\phi_1$  is even more involved. Clearly, if all layer maps $\phi_i$ are injective, then also $f$ is injective. The converse does not hold, however, since~$\phi_i$ only needs to be injective on the image $(\phi_{i-1} \circ \cdots \circ \phi_1)(\R^d)$ for all $i \in [\ell]$. Hence, it is unclear whether the problem is still contained in \FPT when parameterized by $d$.
 
	\section{Verification and Surjectivity}
    \label{sec:surjectivity}
    In this section, we study the network verification task and prove hardness for a very general class of input sets. To this end, we call a sequence $S=(S_d)_{d\in \N}$ of subsets $S_d\subseteq \R^d$ \emph{reasonable}
    if there exists an algorithm that on input~$k\in\N$ computes in $\poly(k)$ time a $d \in \N,\mathbf{z} \in \R^d$ and  a $k$-dimensional affine space $A$ in $\R^d$ (given by a basis and a translation vector) such that there is a $k$-dimensional ball with center $\mathbf{z}$ contained in $S_d$ that affinely spans $A$.
 For a sequence $S=(S_d)_{d\in \N}$ of reasonable sets and a polyhedron $Q_{\mathbf{t}} \coloneqq \{\mathbf{x} \in \R^m \mid x_i \leq t_i$ for all $i \in [m]\}$ where $\mathbf{t} \in \R^m$, the task to decide for a given $\ell$-layer ReLU network $f\colon \R^d \to \R^m$ whether $f(S_d) \subseteq Q_{\mathbf{t}}$ we call \textsc{$\ell$-Layer ReLU $(S,Q_{\mathbf{t}})$-Verification}. 

 Even though the definition of a reasonable sequence is a bit involved, they include a wide range of set-sequences. For example, they trivially include all sequences of sets that contain an open neighborhood of the origin and therefore in particular all sequences of full-dimensional polyhedra containing the origin in the interior. This shows, in comparison to previous existing \coNP-hardness proofs, that the hardness of verification does not rely on the complexity of the polyhedron in the input domain but is intrinsic to the ReLU networks themselves.
   
      To show that the problem is \coNP-hard, it suffices to show hardness for one hidden layer and one-dimensional output and hence in the remaining part we study the following decision problem for an arbitrary $t \in \R$.

       \problemdef{2-Layer ReLU $(S,t)$-Verification}
		{{Matrices $\mathbf{W}_1 \in \R^{n\times d},\mathbf{W}_2 \in \R^{1 \times n}, \mathbf{b}_1 \in \R^{n}$, and $b_2 \in \R$.}}
		{Is $\mathbf{W}_{2}\cdot\phi_{\mathbf{W}_1,\mathbf{b}_1}(\mathbf{x})+b_2 \leq t$ for all $\mathbf{x} \in S_d$?}

      We will see in Theorem~\ref{prop:verification_hard} that the hardness of \textsc{2-Layer ReLU $(S,t)$-Verification} stems from the hardness of deciding whether a ReLU network computes a map that attains a positive output value. Thus, we will first show that the \ReLUPos problem is \NP-complete: Given Matrices $\mathbf{W}_1 \in \R^{n\times d},\mathbf{W}_2 \in \R^{1 \times n}$, is there an $\mathbf{x} \in \R^d$ such that $\mathbf{W}_2\cdot\phi_{\mathbf{W}_1}(\mathbf{x}) > 0$?

        Afterwards, we give a reduction from the complement of \ReLUPos to \textsc{2-Layer ReLU $(S,t)$-Verification} implying that \textsc{2-Layer ReLU $(S,t)$-Verification} and thus \textsc{$\ell$-Layer ReLU $(S,Q_\mathbf t)$-Verification} are \coNP-hard.

        \subsection{NP-Completeness of \ReLUPos}
        \label{section:surjectivity_NP_hard}
        
        We prove \NP-completeness of \ReLUPos via reduction from \textsc{Positive Cut}. For a graph $G=(V,E)$, we denote by $E(S,V\setminus S) \coloneqq \{\{u,v\} \in E \mid u \in S, v \notin S\}$ the edges in the cut induced by $S \subseteq V$. The \textsc{Positive Cut} problem is to decide for a given graph $G=(V,E)$ with edge weights $w\colon E \to \Z$ whether there is a subset $S \subseteq V$ such that $\sum_{e \in  E(S,V\setminus S)} w(e) > 0$, which we prove to be \NP-complete in Appendix~\ref{proof:poscut}. A detailed proof of the following theorem is given in Appendix~\ref{app:thm:Posi-NP}.

	\begin{theorem}\label{thm:Posi-NP}
		\ReLUPos is \NP-complete.
	\end{theorem}
 
    \begin{proofwithcaption}{Proof sketch.}
           The problem is contained in NP since, by \Cref{observation:linearoncells},
          it is sufficient to identify an extreme ray~$\rho$ as a certificate for positivity.
		
		To prove \NP-hardness, we reduce from \textsc{Positive Cut}.  Given a weighted graph $(G=(V,E),w)$ with~$V=[n]$, we define a 2-layer ReLU neural network $f\colon\R^{n} \to\R$, where we have two hidden neurons for every $e \in E$ and the output weights of the neurons are the corresponding weights of the edges. More precisely,
		\[f(\mathbf x)\coloneqq\sum_{\{i,j\}\in E} w(\{i,j\}) \cdot ([x_i-x_j]_+ + [x_j-x_i]_+).\]
 The weights of the cuts are now stored as function values on certain vectors. For a subset $S \subseteq V$, let $\mathbf r_S \coloneqq  \sum_{i \in S} \mathbf e_i \in\R^{n}$ and $\mathbf{r}_S' \coloneqq - \sum_{i \in V\setminus S} \mathbf e_i \in\R^n$. It follows easily that \linebreak \mbox{$
  f(\mathbf r_S)=f(\mathbf{r}_S')=\sum\limits_{\{i,j\}\in E(S,V \setminus S)} w(\{i,j\}).
  $}
  We proceed by showing that every $\mathbf{x}\in\R^n$ is the conic combination of such $\mathbf{r}_S$ and $\mathbf{r}_S'$. Hence, there is an $\mathbf{x}$ with $f(\mathbf{x}) > 0$ if and only if there is an $S \subseteq V$ such that $f(\mathbf r_S)=\sum_{e \in  E(S,V\setminus S)} w(e) > 0$, proving the correctness of the reduction.
    \end{proofwithcaption}

\subsection{\coNP-hardness of Verification}
 We continue with the \coNP-hardness for network verification based on the \NP-hardness of \ReLUPos. 
 \begin{theorem}
        \label{prop:verification_hard}
		For every reasonable set-sequence~$S$ and $t \in \R$, it holds that \textsc{2-Layer ReLU $(S,t)$-Verification} is \coNP-hard. 
\end{theorem}
        \begin{proofwithcaption}{Proof sketch.}
        In Appendix~\ref{app:verification_hard}, we provide a reduction from the complement of \ReLUPos that we sketch here.
 Let $f \colon \R^k \to \R$ be the function computed by an instance of \ReLUPos and let the dimension $d$, the affine space $A$ and the point $\mathbf{z} \in \R^d$ be the output of the algorithm that exists due to the fact that $(S_d)_{d\in \N}$ is a reasonable sequence of sets. 
 Let~$T \colon \R^d \to \R^d$ be the affine map that maps $\mathbf{z}$ to the origin composed with the affine map that projects orthogonal to $A$ and afterwards maps isomorphic to $\R^k$. Then, we prove that $\mathbf{x} \mapsto (f \circ T)(\mathbf{x}) + t$ is the function computed by a ``yes''-instance of \textsc{2-Layer ReLU $(S,t)$-Verification} if and only if~$f$ has no positive output value.
        \end{proofwithcaption}


     \Cref{prop:verification_hard} also implies that $\max_{\mathbf x \in S_d} f(\mathbf x)$ cannot be constant-factor approximated.
    \begin{corollary}
    \label{cor:constapproximation_verification}
    For every reasonable set-sequence $S$, there is no $C >0$ such that there exists a polynomial-time algorithm that, given a 2-layer ReLU neural network $f\colon\R^d \to\R$, computes a value  $a(f) \in \R$ such that $a(f) \leq \max_{\mathbf x \in S_d} f(\mathbf x) \leq C\cdot a(f)$ unless P $=$ \NP.
    \end{corollary}
    \begin{proof}
       Assume that for some reasonable set-sequence $S$, there is a $C>0$ and such an algorithm. Then, we have that $a(f)>0$ if and only if there exists an $\mathbf x \in S_d$ such that $f(\mathbf x) >0$. Hence, by \Cref{prop:verification_hard} (choosing $t=0$), it follows that P $=$ \coNP and thus P $=$ \NP.
    \end{proof}
    \vspace{-0.5cm}
\subsection{Surjectivity}
              
    In this section, we study the task to decide whether a given ReLU network $f \colon \R^d \to \R$ with $\ell$ layers computes a surjective map. We will call this task \textsc{$\ell$-layer ReLU Surjectivity}.
    We start with some simple observations characterizing surjectivity; see Appendix~\ref{app:lemma:charsurjectivity} for a proof. 

        \begin{lemma}\label{lemma:charsurjectivity} For a ReLU network $f=\mathbf{W}_{\ell}\circ\phi_{\mathbf{W}_{\ell-1},\mathbf{b}_{\ell-1}} \circ \cdots \circ \phi_{\mathbf{W}_{1},\mathbf{b}_{1}}+b_{\ell}$ we denote by $f_0 \coloneqq \mathbf{W}_{\ell} \circ \phi_{\mathbf{W}_{\ell-1}} \circ \cdots \circ \phi_{\mathbf{W}_{1}}$ the corresponding ReLU network without biases. Then the following holds:
        \begin{enumerate}[a)]
        \item \label{lemma:char_surj_1} $f_0$ is surjective if and only if there are $\mathbf{v}^+,\mathbf{v}^- \in \R^d$ such that $f_0(\mathbf{v}^+)>0$ and  $f_0(\mathbf{v}^-)<0$.
        \item \label{lemma:char_surj_2}
        $f$ is surjective if and only if $f_0$ is surjective.
        \item \label{lemma:char_surj_3} The map $f$ is surjective if and only if there exist two ray generators $\mathbf{r}^+, \mathbf{r}^-$ of two rays $\rho^+,\rho^- \in \Sigma_{f_0}$ such that $f_0(\mathbf{r}^+) > 0$ and $f_0(\mathbf{r}^-) < 0$. 
          \end{enumerate}
        \end{lemma}
        \Cref{lemma:charsurjectivity}\,\ref{lemma:char_surj_2}) implies that we can assume that the ReLU network has no biases without loss of generality.
        Furthermore, \Cref{lemma:charsurjectivity}\,\ref{lemma:char_surj_3}) implies an exponential-time algorithm for \textsc{$\ell$-Layer ReLU Surjectivity}, since $\Sigma_{f_0}$ contains at most $O(\binom{(\prod_{i=1}^\ell n_i)^d}{d-1})$ many rays. In particular, for $\ell=2$, the fan~$\Sigma_{\mathbf W_1}$ contains at most~$\min \{2n_1^{d-1} ,2^{n_1+1}\}$ rays; that is, the problem is in \XP when parameterized by~$d$ and in \FPT when parameterized by~$n_1$. Moreover, since two ray generators $\mathbf{r}^+, \mathbf{r}^-$ of two rays $\rho^+,\rho^- \in \Sigma_{f_0}$ serve as a certificate, we obtain the following proposition.
    \begin{proposition}
      For every $\ell \in \N$, it holds that \textsc{$\ell$-layer ReLU Surjectivity} is in \NP .
    \end{proposition}
    Again, to show that the problem is \NP-hard, it suffices to consider one hidden layer and one-dimensional output.
    Hence, in the remainder we consider \textsc{2-Layer ReLU Surjectivity}, where we have one hidden layer with~$n$ ReLU-neurons with weights~$\mathbf W_1\in\R^{n\times d}$ and an output layer with one output neuron with weights~$\mathbf W_2\in\R^{1\times n}$ without activation. The network then computes the map $f\colon \R^d\to\R$ with $f(\mathbf x)\coloneqq\mathbf{W}_2\cdot\phi_{\mathbf{W}_1,\mathbf{b}}(\mathbf{x})$.

        In fact, to decide surjectivity, it is actually enough to find one (say the positive) ray generator since it is easy to find some point~$\mathbf x$ where~$f$ is non-zero (we argue below that w.l.o.g.~$f(\mathbf x)<0$). To find~$\mathbf x$, we choose an arbitrary $d$-dimensional cone~$C$ and determine whether~$\mathbf{W}_C$ is the zero map or not. In the former case, one can pick an arbitrary $d$-dimensional cone~$C'$ that shares a facet (i.e., a $(d-1)$-dimensional face) with~$C$. Then, $\mathbf{W}_{C'} \neq \mathbf{0}$ since the neuron defining the facet must be active in~$C'$. Hence, one finds a point~$\mathbf{x} \in C'$ with~$f(\mathbf{x}) \neq 0$. We give a more detailed proof of the following lemma in Appendix~\ref{app:lemma:polytime_zeromap}.
    \begin{lemma}
        \label[lemma]{lemma:polytime_zeromap}
        One can check in polynomial time whether $f = 0$, and otherwise find a point~$\mathbf{x}^* \in \R^d$ such that $f(\mathbf{x}^*) \neq 0$.
    \end{lemma}
        \Cref{lemma:polytime_zeromap} implies that \textsc{2-Layer ReLU Surjectivity} is polynomially equivalent to \ReLUPos (if~$f(\mathbf{x}^*)>0$, then replace~$\mathbf{W}_2$ with~$-\mathbf{W}_2$ such that~$f(\mathbf{x}^*)<0$) and hence we obtain the following corollary.
 \begin{corollary}\label{cor:Surj-hard}
 For every $\ell \geq 2$, it holds that \textsc{$\ell$-Layer ReLU Surjectivity} is \NP-complete.
 \end{corollary}
     Clearly, the above result implies \NP-hardness for the general case where the output dimension~$m$ is part of the input.
     It is unclear, however, whether containment in \NP also holds for larger output dimension.\footnote{We believe that the problem might be~$\Pi^\text{P}_2$-complete for~$m\ge 2$.}

    \subsection{Zonotope Formulation}
    
    We conclude this section with an alternative formulation of \ReLUPos based on a duality of positively homogeneous convex piecewise linear functions and polytopes.
    Interestingly, this yields a close connection to zonotope problems which arise in areas such as robotics and control.
        Let $\mathcal{F}_d$ be the set of convex piecewise linear functions from $\R^d$ to $\R$ and let $\mathcal{P}_d$ be the set of polytopes in~$\R^d$.
        For every $f  \in \mathcal{F}_d$, there are $\{\mathbf{a}_i \in \R^d\}_{i \in I}$ such that  $f(\mathbf x)=\max_{i \in I} \langle\mathbf a_i,\mathbf x\rangle$ and there is a bijection $\varphi \colon \mathcal{F}_d \to \mathcal{P}_d$ given by 
        $
        \varphi\bigl(\max_{i \in I} \langle\mathbf a_i,\mathbf x\rangle\bigr) = \conv{\{\mathbf a_i \mid i \in I\}}$
        where the inverse is the \emph{support function} \mbox{$\varphi^{-1} \colon \mathcal{P}_d \to  \mathcal{F}_d$}  given by 
        $
        \varphi^{-1}(P)(\mathbf x)  = \max \bigl\{ \langle \mathbf x,\mathbf y \rangle \mid \mathbf y \in P \bigr\}.
        $
        Furthermore, $\varphi$ is a semi-ring isomorphism between the semi-rings $(\mathcal{F}_d, \max, +)$ and $(\mathcal{P}_d, \conv,+)$, where $+$ is either the pointwise addition or the Minkowski sum, respectively.\footnote{See, e.g., \cite{pmlr-v80-zhang18i} for more details on this correspondence.}
       
        A \emph{zonotope} is a Minkowski sum of line segments, i.e., given a \emph{generator} matrix $\mathbf{G} \in \R^{n\times d}$, the corresponding zonotope is given by 
        $
        Z(\mathbf{G})\coloneqq \Bigl\{\mathbf{x} \in \R^d \mid \mathbf x \in \sum\limits_{i \in [n]} \conv\{\mathbf 0,\mathbf{g}_i\}\Bigr\}.
        $
        
        \textsc{Zonotope Containment} is the question, given two matrices~$\mathbf G_1\in\R^{n\times d}$, $\mathbf G_2\in\R^{m\times d}$, whether $Z(\mathbf G_1) \subseteq Z(\mathbf G_2)$.
        Using $\varphi$, we prove that \ReLUPos is equivalent to the complement of \textsc{Zonotope Containment} in Appendix~\ref{app:zonotope_containment}.
       	\begin{proposition}
        \label{prop:zonotope_containment}
       	    \ReLUPos is equivalent to \textsc{Zonotope (Non)Containment}.
       	\end{proposition}

       Notably, \citet{KULMBURG202184} already showed that \textsc{Zonotope Containment} is \coNP-hard, which implies our \NP-hardness of \ReLUPos (\Cref{thm:Posi-NP}).
       Nevertheless, we believe that our reduction is more accessible and direct and provides a different perspective on the computational hardness.
       The question whether \textsc{Zonotope Containment} is fixed-parameter tractable with respect to~$d$ is, to the best of our knowledge, open.

	\section{Conclusion}
 \label{section:conclusion}
        We showed the strongest hardness result for network verification known so far and thereby excluded polynomial-time algorithms (in the worst case) for almost all restrictions on the input set. Moreover, we initiate the complexity-theoretic study of deciding the two elementary properties injectivity and surjectivity for functions computed by ReLU networks.
        We exclude polynomial-time algorithms for solving both problems, and prove fixed-parameter tractability for injectivity of a single layer.
        It turned out that surjectivity is a special case of network verification and is also equivalent to zonotope containment.
        Our results build new bridges between seemingly unrelated areas, and yield new insights into the complexity and expressiveness of ReLU neural networks.
        We close with some open questions:
        \begin{compactitem}
        \item[--] Can the running time for \ReLUInj be improved? Or is it possible to prove a lower bound of~$2^{\Omega(d\log d)}$?
        \item[--] What is the (parameterized) complexity of deciding injectivity for a ReLU neural network with two hidden layers? 
          \item[--] Is \textsc{2-Layer ReLU Surjectivity} (or equivalently \textsc{Zonotope Containment}) fixed-parameter tractable with respect to the input dimension~$d$?
          \item[--] How can surjectivity be characterized for output dimension~$m\ge 2$? What is the complexity of the decision problem?
          \item[--] What is the complexity of bijectivity for 2-layer ReLU networks?
        \end{compactitem}
	

\bibliography{ref}
\newpage
\appendix

\section{Appendix}

 \subsection{Appendix to Section 3}
 \label{sec:appendix}
 
 To prove that \AcycDisc is \NP-hard, we reduce from the following problem.
   \problemdef{3-Uniform Hypergraph 2-Coloring}
	{A 3-uniform hypergraph~$H=(V,E)$, that is, $|e|=3$ for all~$e\in E$.}
	{Is there a 2-coloring of the nodes~$V$ such that no hyperedge is monochromatic?}
 \begin{theorem}\label{thm:acyc-NP}
		\AcycDisc is \NP-hard.
	\end{theorem}
 \label{Proof:Acyclic_Np_hard}

	\begin{proof}
		We give a reduction from \textsc{3-Uniform Hypergraph 2-Coloring} which is \NP-hard~\citep{Lovasz73}. Let~$H=(V,E)$ be a 3-uniform hypergraph with~$|V|=n$ and~$|E|=m$. We construct a digraph~$D=(U,A)$ as follows: For each~$i\in\{0,1\}$, we define a node set~$U_i$ with the~$n+2m$ nodes
		\[U_i\coloneqq \{v_i\mid v\in V\} \cup \bigcup_{e\in E}\{e_i,e_i'\}.\]
		Moreover, for each node~$v\in V$, we define a node set~$X_v$ with the $2\deg(v)+2$ nodes
		\[X_v\coloneqq \{x_v,x_v'\}\cup \{x_{v,e,0},x_{v,e,1}\mid e\in E, v\in e\}.\]
		Finally, we define the node set~$Q\coloneqq \{q_{e,i},q_{e,i}'\mid e\in E, i\in\{0,1\}\}$ and let
		\[U\coloneqq U_0 \cup U_1 \cup \bigcup_{v\in V}X_v\cup Q.\]
		The arc set~$A$ is defined as follows: For~$i\in\{0,1\}$, we connect the nodes in~$U_i$ with $2(|U_i|-1)$ arcs to a strongly connected path, that is, for an arbitrary ordering~$b_1,\ldots,b_{|U_i|}$ of the nodes in~$U_i$, we insert the arcs~$(b_j,b_{j+1})$ and~$(b_{j+1},b_j)$ for each~$j\in[|U_i|-1]$. Analogously, we also connect all nodes in~$X_v$ for each~$v\in V$ to a strongly connected path. Moreover, for each~$v\in V$, we insert the cyclic arcs $(v_0,x_v)$, $(x_v,v_1)$, $(v_1,x_v')$, and~$(x_v',v_0)$.
		Finally, for each hyperedge $e=\{u,v,w\}\in E$ and each~$i\in\{0,1\}$, we insert the arcs
		\begin{compactitem}
			\item $(q_{e,i},q_{e,i}')$ and $(q_{e,i}',q_{e,i})$,
			\item $(x_{u,e,i}, e_i)$, $(e_i,q_{e,i})$, $(q_{e,i},e_i')$, and $(e_i',x_{u,e,i})$, and
			\item $(x_{v,e,i},q_{e,i})$, $(q_{e,i},x_{w,e,i})$, $(x_{w,e,i},q_{e,i}')$, and $(q_{e,i}', x_{v,e,i})$.
		\end{compactitem}
		Overall, the constructed digraph~$D$ contains~$O(n+m)$ many nodes and arcs.
		
		For the correctness, assume first that there is a 2-coloring of the nodes of~$H$ such that no hyperedge is monochromatic and let~$V_i\subseteq V$ denote the set of nodes with color~$i$. We construct a solution~$A'\subseteq A$ for~$D$ as follows:
		For each~$v\in V_0$, $A'$ contains the arcs~$(x_v,v_1)$ and~$(v_1,x_v')$, and for each~$v\in V_1$, $A'$ contains the arcs~$(v_0,x_v)$ and~$(x_v',v_0)$.
		Clearly, these arcs are acyclic.
		Further, consider a hyperedge~$e=\{u,v,w\}\in E$. Since~$e$ is not monochromatic, it follows that exactly one of its nodes is colored with one of the two colors, say~0 (the other case is analogous), and the two other nodes have color~1.
		If~$u$ is colored~0, then~$A'$ contains the arcs $(x_{u,e,1},e_1)$, $(e_1',x_{u,e,1})$, $(e_0,q_{e,0})$, and~$(q_{e,0},e_0')$.
		If~$v$ has color~0 (the case where~$w$ has color~0 is analogous), then~$A'$ contains the arcs $(x_{u,e,0},e_0)$, $(e_0',x_{u,e,0})$, $(x_{v,e,1},q_{e,1})$, $(q_{e,1},x_{v,e,1})$, $(x_{w,e,0},q_{e,0}')$, and $(q_{e,0},x_{w,e,0})$.
		It can easily be verified that~$A'$ is acyclic.
		Moreover,~$(U,A\setminus A')$ is not weakly connected since e.g.~all nodes in~$X_v$ with~$v\in V_0$ are disconnected from all nodes in~$U_1$.
		
		Conversely, assume that there is a solution~$A'\subseteq A$ for~$D$. First, observe that in~$D'\coloneqq (U,A\setminus A')$ all nodes in~$U_0$ are weakly connected (since they are strongly connected in~$D$). The same holds for~$U_1$ and for each~$X_v$, $v\in V$. Moreover, for each~$v\in V$, the set~$X_v$ is weakly connected to exactly one of the sets~$U_0$ or~$U_1$.
		To see this, note that~$X_v$ cannot be disconnected from both~$U_0$ and~$U_1$ due to the cycle involving the nodes~$x_v$ and~$x_v'$.
		It follows that also no~$X_v$ can be weakly connected to both~$U_0$ and~$U_1$ since then $D'$ would be weakly connected because also each node in~$Q$ is connected to some~$X_v$ due to its cyclic connections. Hence, we assign each node~$v\in V$ the color~$i\in\{0,1\}$ if and only if~$X_v$ is weakly connected to~$U_i$.
		
		It remains to check that each hyperedge~$e=\{u,v,w\}\in E$ is not monochromatic. Assume the contrary, that is,~$X_u$, $X_v$, and~$X_w$ are weakly connected to (wlog)~$U_0$.
		Then, by construction, also~$q_{e,1}$ and~$q_{e,1}'$ are weakly connected to~$U_0$. Since $X_u$ is not weakly connected to $U_1$ and therefore $A'$ contains $(x_{u,e,1},e_1)$ and $(e_1',x_{u,e,1})$, it follows that~$e_1$ is weakly connected to~$q_{e,1}$, since otherwise $A'$ would not be acyclic. Therefore, $e_1$ is also weakly connected to $U_0$, which yields a contradiction since then~$D'$ would be weakly connected. 
	\end{proof}
\begin{figure}
     \centering
     \begin{tikzpicture}

        \node (U0) at (-1.5,0) {$U_0$};
        \node (u0)  at (1,0) {$u_0$};
        \node (v0) [right of=u0]  {$v_0$};
        \draw (v0) edge [->] (u0);
        \draw (u0) edge [->] (v0);
        \node [right of=v0] (w0) {$w_0$};
        \draw (w0) edge [->] (v0);
        \draw (v0) edge [->] (w0);
        \node[right of=w0] (ldots) {$\ldots$};
        \draw (w0) edge [->] (ldots);
        \draw (ldots) edge [->] (w0);
        \node[right of=ldots] (e0) {$e_0$};
        \node[right of=e0] (e0') {$e_0'$};
        \draw (e0) edge [->] (e0');
        \draw (e0') edge [->] (e0);
        \draw (e0) edge [->] (ldots);
        \draw (ldots) edge [->] (e0);
        \node[right of=e0'] (ldots2) {$\ldots$};
        \draw (e0') edge [->] (ldots2);
        \draw (ldots2) edge [->] (e0');
        \node[ellipse, draw=black, fit=(u0) (w0) (v0) (ldots) (e0) (e0') (ldots2)] (all) {};
        \node[ellipse, draw=black, fill=blue, opacity=0.2,fit=(u0) (w0) (v0) (ldots) (e0) (e0') (ldots2)] (all) {};

        \node (U1) at (-1.5,-6) {$U_1$};
        \node (u1) at (1,-6){$u_1$};
        \node[right of=u1] (v1) {$v_1$};
        \node[right of=v1] (w1) {$w_1$};
        \draw (v1) edge [->] (u1);
        \draw (v1) edge [->] (w1);
        \draw (u1) edge [->] (v1);
        \draw (w1) edge [->] (v1);
        \node[right of=w1] (rdots) {$\ldots$};
        \node[right of=rdots] (e1) {$e_1$};
        \node[right of=e1](e1'){$e_1'$};
        \node[right of=e1'] (rdots2) {$\ldots$};
        \draw (e1) edge [->] (e1');
        \draw (e1') edge [->] (e1);
        \draw (e1) edge [->] (rdots);
        \draw (rdots) edge [->] (e1);
        \draw (rdots) edge [->] (w1);
        \draw (rdots2) edge [->] (e1');
        \draw (w1) edge [->] (rdots);
        \draw (e1') edge [->] (rdots2);
        \node[ellipse, draw=black, fit=(u1) (w1) (v1) (rdots) (e1) (rdots2)] (all) {};
        \node[ellipse, draw=black, fill=red,opacity=0.2, fit=(u1) (w1) (v1) (rdots) (e1) (rdots2)] (all) {};

        \node[shape=circle, draw = black, scale=0.4] (qe0) at (4,-1.5){};
        \node[shape=circle, draw = black, scale=0.4,right of =qe0] (qe0'){};
        \node[right of=qe0'] (Q0) {$Q_{e,0}$};
        \draw (qe0) edge [->] (qe0');
        \draw (qe0') edge [->] (qe0);
        \node[ellipse, draw=black, fit=(qe0) (qe0')] (all) {};

        \node[shape=circle, draw = black, scale=0.4] (qe1) at (4,-4.5){};
        \node[shape=circle, draw = black, scale=0.4,right of =qe1] (qe1'){};
        \node[right of=qe1'] (Q1) {$Q_{e,1}$};
        \draw (qe1) edge [->] (qe1');
        \draw (qe1') edge [->] (qe1);
        \node[ellipse, draw=black, fit=(qe1) (qe1')] (all) {};

         \node[shape=circle, draw = black, scale=0.4] (xv) at (2.5,-3){};
         \node[shape=circle, draw = black, scale=0.4,right of=xv] (xv'){};
         \node[shape=circle, draw = black, scale=0.4,right of=xv'] (xve0){};
         \node[shape=circle, draw = black, scale=0.4,right of=xve0] (xve1){};
         \node[right of =xve1] (ldotsv) {$\ldots$};
         \node (Xv) at (2.5,-3.5) {$X_v$};
         \node[ellipse, draw=black, fit=(xv) (xv') (xve0) (xve1) (ldotsv), inner sep = -0.1mm] (all) {};
         \draw (xv) edge[->] (xv');
         \draw (xv') edge[->] (xv);
         \draw (xve0) edge[->] (xv');
         \draw (xv') edge[->] (xve0);
         \draw (xve0) edge[->] (xve1);
         \draw (xve1) edge[->] (xve0);

        \node[shape=circle, draw = black, scale=0.4] (xu) at (-1.5,-3){};
          \node[shape=circle, draw = black, scale=0.4,right of=xu] (xu'){};
           \node[shape=circle, draw = black, scale=0.4,right of=xu'] (xue0){};
         \node[shape=circle, draw = black, scale=0.4,right of=xue0] (xue1){};
         \node[right of =xue1] (ldotsu) {$\ldots$};
         \node (Xu) at (-1.5,-3.5) {$X_u$};
         \node[ellipse, draw=black, fit=(xu) (xu') (xue0) (xue1) (ldotsu), inner sep = -0.1mm] (all) {};
         \draw (xu) edge[->] (xu');
         \draw (xu') edge[->] (xu);
         \draw (xue0) edge[->] (xu');
         \draw (xu') edge[->] (xue0);
         \draw (xue0) edge[->] (xue1);
         \draw (xue1) edge[->] (xue0);

         \node[shape=circle, draw = black, scale=0.4] (xw) at (6.5,-3){};
         \node[shape=circle, draw = black, scale=0.4,right of=xw] (xw'){};
         \node[shape=circle, draw = black, scale=0.4,right of=xw'] (xwe0){};
         \node[shape=circle, draw = black, scale=0.4,right of=xwe0] (xwe1){};
         \node[right of =xwe1] (ldotsw) {$\ldots$};
         \node(Xw) at (9.5,-3.5) {$X_w$};
         \node[ellipse, draw=black, fit=(xw) (xw') (xwe0) (xwe1) (ldotsw), inner sep = -0.1mm] (all) {};
         \draw (xw) edge[->] (xw');
         \draw (xw') edge[->] (xw);
         \draw (xwe0) edge[->] (xw');
         \draw (xw') edge[->] (xwe0);
         \draw (xwe0) edge[->] (xwe1);
         \draw (xwe1) edge[->] (xwe0);

        \draw (u0) edge [->] (xu);
        \draw (xu) edge [->] (u1);
        \draw (u1) edge [->] (xu');
        \draw (xu') edge [->] (u0);



        \draw (xue0) edge [->] (e0);
        \draw (e0') edge [->] (xue0);
        \draw (e0) edge [->] (qe0);
        \draw (qe0) edge [->] (e0');

        \draw (xve0) edge [->] (qe0);
        \draw (qe0') edge [->] (xve0);
        \draw (qe0) edge [->] (xwe0);
        \draw (xwe0) edge [->] (qe0');


        \draw (xue1) edge [->] (e1);
        \draw (e1') edge [->] (xue1);
        \draw (e1) edge [->] (qe1);
        \draw (qe1) edge [->] (e1');
        
        \draw (xve1) edge [->] (qe1);
        \draw (qe1') edge [->] (xve1);
        \draw (qe1) edge [->] (xwe1);
        \draw (xwe1) edge [->] (qe1');

    \end{tikzpicture}   
    \caption{The encoding of one hyperedge $
    \{u,v,w\} \in E$ in the digraph $D$. The cyclic directed paths connecting $X_v$ respectively $X_w$ with $U_0$ and $U_1$ are not drawn in order to not overload the figure.}
 \end{figure}

 We remark that our reduction implies a running time lower bound based on the Exponential Time Hypothesis\footnote{The Exponential Time Hypothesis asserts that \textsc{3-SAT} cannot be solved in $2^{o(n)}$ time where~$n$ is the number of Boolean variables in the input formula (\cite{IP01}).} (ETH).
	As discussed by \citet{KL16}, there is no algorithm solving a \textsc{3-Uniform Hypergraph 2-Coloring}-instance~$(V,E)$ in $2^{o(|E|)}$ time assuming ETH.
	Our polynomial-time reduction in the proof of \Cref{thm:acyc-NP} constructs a digraph~$D$ of size~$O(|V|+|E|)$. Notice that~$O(|V|+|E|)\subseteq O(|E|)$ since we can assume $|V|\le 3|E|$ (isolated nodes can trivially be removed).
	Hence, any algorithm solving \AcycDisc in $2^{o(|D|)}$ time would imply a~$2^{o(|E|)}$-time algorithm for \textsc{3-Uniform Hypergraph 2-Coloring}.

    \begin{corollary}\label{cor:AD-lower}
		\AcycDisc cannot be solved in~$2^{o(|D|)}$ time unless the ETH fails.
	\end{corollary}
 \begin{proofwithcaption}{Proof of Theorem~\ref{thm:ReLUInj-NPhard}.}
 \label{proof:ReLUInj-NPhard}
		Containment in \NP is easy: The set~$I_C$ of a cell~$C$ with~$\rank(\mathbf{W}_C)<d$ serves as a certificate.
		
		For the \NP-hardness, we reduce from \AcycDisc which is \NP-hard by \Cref{thm:acyc-NP}.
		For a given digraph $D = (V=\{v_1,\ldots,v_n\},A=\{a_1,\ldots,a_m\})$, we construct the matrix $\mathbf{W} \in \R^{m\times (n-1)}$ as follows:
		For every arc $a_\ell=(v_i,v_j)\in A$, we add the row vector $\mathbf{w}_\ell\in \R^{n-1}$ with
        \[(\mathbf w_\ell)_k \coloneqq \begin{cases}1, &k=i\\-1, &k=j\\0, &\text{ else}\end{cases}.\]
		
        For the correctness, assume first that there is a solution~$A'\subseteq A$ for~$D$.
		Let~$I'\coloneqq\{i\mid a_i\in A'\}\subseteq [m]$.
		We claim that there is a cell~$C\in\mathcal{C}_\mathbf{W}$ with~$I_C\subseteq [m]\setminus I'$.
		To see this, let~$\pi \in\mathcal S_n$ be a permutation with~$\pi(i) < \pi(j)$ for each $(v_i,v_j)\in A'$.
		Such a permutation exists since~$A'$ is acyclic.
		Let~$x_n\coloneqq 0$ and let~$\mathbf{x}\in\R^{n-1}$ be such that
		$x_{\pi(1)}< x_{\pi(2)}< \cdots < x_{\pi(n)}$.
		Then, no neuron~$i\in I'$ corresponding to an arc~$a_i=(v_j,v_\ell)\in A'$ is active at~$\mathbf{x}$ since
		\[\inner{\mathbf{w}_i,\mathbf{x}} = \begin{cases}x_j - x_\ell, &j\neq n \wedge \ell\neq n\\ -x_\ell, &j=n\\
        x_j, &\ell=n\end{cases}< 0.\]
		Hence,~$\mathbf{x}$ is contained in some cell~$C$ with~$I_C\subseteq [m]\setminus I'$.
		Now, since~$(V,A\setminus A')$ is not weakly connected, it follows that~$\rank(\mathbf{W}_C)<n-1$.
		To see this, note that there must be a node~$v_i\in V$ that is not weakly connected to~$v_n$, that is, there is no undirected path from~$v_i$ to~$v_n$.
		If~$\rank(\mathbf{W}_C)=n-1$, then there exists a linear combination $\sum_{j\in I_C}c_j\mathbf{w}_{j} = \mathbf{e}_i$ of the $i$-th unit vector in~$\R^{n-1}$.
		But this implies the existence of an undirected path from~$v_i$ to~$v_n$ corresponding to some arcs in~$\{a_j\mid c_j\neq 0\}$, which yields a contradiction.
		To see this, note first that $(\mathbf{e}_i)_i=1$ implies that~$c_j\neq 0$ for some~$j\in I_C$ such that~$(\mathbf{w}_{j})_i\neq 0$.
		Clearly, $a_j=(v_i,v_n)$ or~$a_j=(v_n,v_i)$ is not possible.
		Hence,~$a_j$ must be an arc between~$v_i$ and some~$v_k\neq v_n$.
		But then, we have~$(c_j\mathbf{w}_{j})_k\neq 0$, whereas~$(\mathbf{e}_i)_k=0$.
		Therefore, there exists another~$j'\in I_C$ such that~$c_{j'}\neq 0$ and~$(\mathbf{w}_{j'})\neq 0$.
		Again, it is not possible that~$a_{j'}=(v_k,v_n)$ or~$a_{j'}=(v_n,v_k)$.
		However, since~$I_C$ is finite, repeating this argument yields a contradiction.
		
		For the reverse direction, let~$C\in\mathcal{C}_\mathbf{W}$ be a cell with~$\rank(\mathbf{W}_C)<n-1$.
		Then, there exists a point~$\mathbf{x}\in C$.
		Now, let~$i\in [m]\setminus I_C$ be a neuron corresponding to arc~$a_i=(v_j,v_\ell)$ that is not active at~$\mathbf{x}$.
		Then, this implies that $x_j - x_\ell < 0$ if~$j < n$ and~$\ell <n$.
		If~$j=n$, then this implies~$x_\ell > 0$ and if~$\ell=n$, then this implies $x_j < 0$.
		Hence, $(V,A')$ with~$A'\coloneqq \{a_i\mid i\in [m]\setminus I_C\}$ is acyclic since any cycle would lead to a contradiction.
		Now, we claim that $(V, A\setminus A')$ cannot be weakly connected.
		Otherwise, there exists an undirected path from each~$v_i$, $i\in[n-1]$ to~$v_n$.
		Let~$v_i=v_{j_0},v_{j_1},\ldots,v_{j_t}=v_n$ be the nodes of such a path along the arcs~$a_{\ell_1},\ldots,a_{\ell_t}$.
		Then, there exists the linear combination~$\sum_{k=1}^t c_k\mathbf{w}_{\ell_k}= \mathbf{e}_i$, where~$c_k\coloneqq 1$ if~$a_{\ell_k}=(v_{j_{k-1}},v_{j_k})$ and~$c_k\coloneqq -1$ if~$a_{\ell_k}=(v_{j_k},v_{j_{k-1}})$.
		But this implies $\rank(\mathbf{W}_C)=n-1$, which yields a contradiction.
	 \end{proofwithcaption}

 \subsection{Appendix to Section 4}
 \begin{proofwithcaption}{Proof of Lemma~\ref{lem:conecover}}
 \label{app:lem:conecover}

 We use strong duality of linear programming.
		To that end, let $\mathbf{A} \in \R^{m \times d}, \mathbf p \in \R^m$ such that $C =  \{\mathbf{x} \in \R^d \mid \inner{\mathbf{a},\mathbf{x}} \geq \mathbf{p}\}$ and let $\mathbf{W} \coloneqq (\mathbf{w}_1, \ldots ,\mathbf{w}_n)^T \in \R^{n \times d}$ and $\mathbf{b} \in \R^n$. Let $\varepsilon > 0$, then, since $C \subseteq \bigcup\limits_{i=1}^n H_{\mathbf{w}_i,b_i}^+$, it follows that the set \[ \{\mathbf x \in \R^d \mid -\inner{\mathbf{w},\mathbf{x}} \geq -(\mathbf{b} + \varepsilon\mathbf{1}), \inner{\mathbf{a},\mathbf{x}} \geq \mathbf{p}\} \subseteq  \{\mathbf x \in \R^d \mid -\inner{\mathbf{w},\mathbf{x}} > -\mathbf{b}, \inner{\mathbf{a},\mathbf{x}} \geq \mathbf{p}\}\] is empty and hence the following linear program does not admit a feasible solution.
		\begin{equation*}
			\begin{matrix}
				\displaystyle \min_\mathbf x & 0  \\
				\textrm{s.t.} & \inner{\mathbf{a},\mathbf{x}} & \geq & \mathbf{p}  \\
				& -\inner{\mathbf{w},\mathbf{x}} & \geq & -(\mathbf{b} + \varepsilon\mathbf{1})  \\
			\end{matrix}
		\end{equation*}
		By strong duality, the dual linear program
		\begin{equation*}
			\begin{matrix}
				\displaystyle \max_{(\mathbf{y},\mathbf{z})} & -(\mathbf{b} + \varepsilon\mathbf{1})^T \mathbf{y} +\mathbf{p}^T \mathbf{z} \\
				\textrm{s.t.} & \mathbf{z}^T\mathbf{A} - \mathbf{y}^T\mathbf{W}   & = & \mathbf{0} \\
				& \mathbf{y}, \mathbf{z}  & \geq & \mathbf{0}  \\
			\end{matrix}
		\end{equation*}
		has either no feasible solution or its objective value is unbounded. Since $\mathbf{y}=\mathbf{0}$ and $\mathbf{z} = \mathbf{0}$ yields a feasible solution, the latter is the case.
		In particular, there is a ray $\rho$ of the cone $\{(\mathbf{y},\mathbf{z}) \in \R^{n+m} \mid \mathbf{z}^T\mathbf{A} - \mathbf{y}^T\mathbf{W}  =  \mathbf{0},\mathbf{y}, \mathbf{z} \geq  \mathbf{0}\}$ such that the objective value is unbounded on~$\rho$.
		The dimension of the subspace  $\{(\mathbf y,\mathbf z) \in \R^{n+m} \mid \mathbf{z}^T\mathbf{A} - \mathbf{y}^T\mathbf{W}    =  \mathbf{0}\}$ is at least $n+m-d$. Therefore, $\rho$ as a $1$-dimensional subspace lies in the intersection of at least $n+m-d-1$ many hyperplanes of the form $\{(\mathbf y,\mathbf z) \in \R^{n+m}\mid (\mathbf y,\mathbf z)_i = 0\}$ and hence in the intersection of at least $n-d-1$ many hyperplanes of the form $\{(\mathbf y,\mathbf z) \in  \R^{n+m} \mid y_i = 0\}$.
		Let $B \subseteq [n]$ be the set of size at least $n-d-1$ such that $\rho \subseteq  \bigcap_{i \in B}\{(\mathbf y,\mathbf z) \in  \R^{n+m} \mid y_i = 0\}$ (this can be computed in polynomial time; see, e.g., \cite[Corollary 14.1g]{ilp_theory}) and $A \coloneqq [n] \setminus B$ its complement.
		
		It follows that $|A| \leq d+1$ and that the objective value of the following LP is still unbounded.
		\begin{equation*}
			\begin{matrix}
				\displaystyle \max_{(\mathbf{y},\mathbf{z})} & -(\mathbf{b} + \varepsilon\mathbf{1})^T \mathbf{y} +\mathbf{p}^T \mathbf{z}  \\
				\textrm{s.t.} & \mathbf{z}^T\mathbf{A} - \mathbf{y}_{A}^T\mathbf{W}_{A}   & = & \mathbf{0} \\
				& \mathbf{y}, \mathbf{z}  & \geq & \mathbf{0}  \\
			\end{matrix}
		\end{equation*}
		Again, by strong duality, this implies that the set \[ \{\mathbf x \in \R^d \mid -\mathbf{W}_{A}\mathbf{x} \geq -(\mathbf{b}_A + \varepsilon\mathbf{1}_A), \inner{\mathbf{a},\mathbf{x}} \geq \mathbf{p}\}\] is empty. Since $\varepsilon > 0$ was arbitrary, it follows that also the set \[ \{\mathbf x \in \R^d \mid -\mathbf{W}_{A}\mathbf{x} > \mathbf{b}_A, \inner{\mathbf{a},\mathbf{x}} \geq \mathbf{p}\}\] is empty, which means that $C \subseteq \bigcup\limits_{i \in A} H_{\mathbf{w}_i,b_i}^+$, proving the claim.
        
	\end{proofwithcaption}

 \paragraph{Proof of \Cref{thm:fpt}}   
\label{app:thm:fpt}
	\begin{lemma}\label{lem:correctness}
		\Cref{Layerinjectivity} is correct.
	\end{lemma}
\begin{proof}
		Assume that the algorithm outputs some~$\mathbf x\in\R^d$. If~$\mathbf x$ was returned in \Cref{trivial1} or \Cref{trivial2}, then this is clearly correct since either no cell has rank~$d$ or there is a rank-0 cell.
		If~$\mathbf x$ was returned by some call of FindCell (\Cref{FindCell}) in \Cref{returnX}, then
		this is correct since~$\rank(W_\mathbf x)\le\rank(\{\mathbf c_i\}_{i \in [n]}) < d$, where $W_\mathbf x \coloneqq \{\mathbf w_i\in W\mid \mathbf w_i\mathbf x +b_i \ge 0\}$.
		The second inequality holds since otherwise the algorithm would have returned ``No'' in \Cref{checkNo}.
		For the first inequality, we first observe the invariant that $W \setminus \spn(C)\subseteq M\subseteq W$ holds at any time during execution of FindCell.
		This is clear for the initial call of FindCell in \Cref{find} in \Cref{Layerinjectivity}.
		Also, within FindCell the property holds after \Cref{filter} and for the recursive calls in \Cref{recurse}.
		Now, $\rank(W_\mathbf x)\le \rank(C)$ holds since~$\mathbf m_i\mathbf x +b_i < 0$ for all $i\in[m]$ implies that~$W_\mathbf x\cap M= \emptyset$, and thus, by our invariant, we have~$C \subseteq W_\mathbf x \subseteq \spn(C)$.
		
		For the opposite direction, assume that there is an $\mathbf{x} \in \R^d$ such that $k \coloneqq \rank(W_\mathbf{x}) < d$.
		If $k=0$, then \Cref{Layerinjectivity} correctly returns some point of a rank-0 cell in \Cref{trivial2}.
		If~$k > 0$, then $W_\mathbf x$ contains at least one vector~$\mathbf w_i$.
		We claim that the call of FindCell in \Cref{find} will correctly return some point from a cell with rank at most~$k$.
		To this end, we show that FindCell($C$, $P$, $M$, $B$) (\Cref{FindCell}) returns a correct point whenever~$C\subseteq W_\mathbf x$.
		Clearly, if $\rank(C)=k$, then~$\mathbf x$ satisfies the conditions in \Cref{returnX} since $C\subseteq W_\mathbf x$, and thus $W_\mathbf x\subseteq \spn(C)$ while $M \subseteq W\setminus \spn(C)$ (due to \Cref{filter} and the invariant on~$M$).
		Hence, FindCell returns a correct point in this case.
		
		Now consider the case $\rank(C) < k$.
		If some point is returned in \Cref{returnX}, then this is correct (as already shown above).
		If no point satisfies the conditions in \Cref{returnX}, then it holds $\{\mathbf x\mid \forall i\in[n]:\mathbf c_i\mathbf x +p_i\ge 0\}\subseteq \bigcup_{i\in I}H^+_{\mathbf m_i}$ and
		by \Cref{lem:conecover}, we can compute the set~$A$ in \Cref{cover}.
		Note that $W_\mathbf x \cap \{\mathbf m_i\mid i \in A\}\neq \emptyset$.
		Hence, for at least one of the recursive calls in \Cref{recurse}, it holds that $C\cup\{\mathbf m_i\} \subseteq W_\mathbf x$.
		Moreover, $\rank(C\cup\{\mathbf m_i\})=\rank(C)+1$ since~$\mathbf m_i\not\in\spn(C)$ (due to \Cref{filter}).
		Hence, by induction, this call will return a correct point.
	\end{proof}

 	\begin{lemma}\label{lem:runtime}
		\Cref{Layerinjectivity} runs in~$O((d+1)^d\cdot \poly(S))$ time, where~$S$ denotes the input size.
	\end{lemma}
\begin{proof}
		Let~$S$ be the bit-length of~$W$.
		Clearly, \Cref{trivial1} and \Cref{trivial2} can be done in~$\poly(S)$ time via linear programming. As regards the running time of FindCell, note first that the recursion depth is at most~$d$ since every recursive call increases the rank of~$C$ (as already discussed in the proof of \Cref{lem:correctness}) and the recursion terminates when rank~$d$ is reached.
		Moreover, each call of FindCell branches into at most~$d+1$ recursive calls, that is, the search tree has size at most~$(d+1)^d$.
		Since all other computations within FindCell can be done in~$\poly(S)$ time (using linear programming and \Cref{lem:conecover}), we obtain the desired running time.
	\end{proof}

\subsection{Appendix to Section 5}
 \begin{proofwithcaption}{Proof of \Cref{prop:verification_hard}}
 \label{app:verification_hard}

 We reduce from the complement of \ReLUPos.
        Let $\mathbf{W}_1 \in \R^{n\times k},\mathbf{W}_2 \in \R^{1 \times n}$ be an instance of \ReLUPos. 
        Let the dimension $d$, the affine space $A$ (as a basis) and the point $\mathbf{z} \in \R^d$ be the output of the algorithm (on input $k$) that exists due to the fact that $(S_d)_{d\in \N}$ is a reasonable sequence of sets. More precisely, let $\varepsilon > 0$ be chosen such that \[B=B_\varepsilon(\mathbf{z})\coloneqq\{\mathbf x\in A \mid \|\mathbf x - \mathbf z\|_2<\varepsilon\} \subseteq S_d.\]
        Let $P \colon \R^d \to A$ be the orthogonal projection to $A$ and $T \colon A \to \R^k$ an isometric isomorphism obtained by mapping the normalized basis of $A$ to the standard basis of $\R^k$.
        The composition $(T \circ P) \colon \R^d \to \R^k$ is an affine map and let it be given by a matrix $\mathbf{A} \in \R^{k \times d}$ and a vector $\mathbf{b} \in \R^k$. 
        Then, $\widetilde{\mathbf{W}}_2 \coloneqq \mathbf{W}_2,\widetilde{\mathbf{W}}_1 \coloneqq \mathbf{W}_1 \cdot \mathbf{A},\widetilde{\mathbf{b}}_1 \coloneqq \mathbf{W}_1 \cdot \mathbf{b} - \mathbf{W}_1 \cdot \mathbf{A} \cdot \mathbf{z}, \widetilde{b}_2 \coloneqq  t$ form an instance of \textsc{2-Layer ReLU $(S,t)$-Verification} and it holds that 
        \begin{equation}
		\label{eq:verification}
		\mathbf{W}_2 \cdot \phi_{\mathbf{W}_1}(\inner{\mathbf{a},\mathbf{x}} + \mathbf{b}) > 0 \iff \widetilde{\mathbf{W}}_2\cdot\phi_{\widetilde{\mathbf{W}}_1,\widetilde{\mathbf{b}}_1}(\mathbf{z}+\mathbf{x})+\widetilde{b}_2 > t.
        \end{equation}

        For the correctness of the reduction, assume that there is a $\mathbf{y} \in \R^k$ such that $\mathbf{W}_2 \cdot \phi_{\mathbf{W}_1}(\mathbf{y}) > 0$. Then, by the positive homogeneity of the map $\mathbf{x} \mapsto \mathbf{W}_2 \cdot \phi_{\mathbf{W}_1}(\mathbf{x})$, we also have that $\mathbf{W}_2 \cdot \phi_{\mathbf{W}_1}(\mathbf{y}') > 0$ for $\mathbf{y}' \coloneqq \frac{\varepsilon}{2} \frac{\mathbf y}{\|\mathbf{y}\|}$. Since $T$ is an isometric isomorphism, there is an $\mathbf{x} \in A$ with $\|\mathbf{x}\| = \frac{\varepsilon}{2}$ such that $T(\mathbf{x}) =\mathbf{y}'$. Hence, $\mathbf{z} + \mathbf{x} \in B\subseteq S_d$ and \Cref{eq:verification} implies that $\widetilde{\mathbf{W}}_2\cdot\phi_{\widetilde{\mathbf{W}}_1,\widetilde{\mathbf{b}}_1}(\mathbf{z}+\mathbf{x})+\widetilde{b}_2> t$.
		
		Conversely, if there is an $\mathbf{x} \in S_d$ such that $\widetilde{\mathbf{W}}_2\cdot\phi_{\widetilde{\mathbf{W}}_1,\widetilde{\mathbf{b}}_1}(\mathbf{x})+\widetilde{b}_2 > t$, then \Cref{eq:verification} implies that $\mathbf{W}_2 \cdot \phi_{\mathbf{W}_1}(\mathbf{A}(\mathbf{x}-\mathbf{z})+\mathbf{b}) > 0$, concluding the proof.
 \end{proofwithcaption}
 
\begin{proofwithcaption}{Proof of Lemma~\ref{lemma:charsurjectivity}}
\label{app:lemma:charsurjectivity}
\ref{lemma:char_surj_1}) If there is a $\mathbf{v}^+\in\R^d$ with $f_0(\mathbf{v}^+) = a > 0$, then for all $b \in [0,\infty)$ it holds that $f_0(\frac{b}{a}\mathbf{v}^+) = b$ due to positive homogeneity of~$f_0$ (analogously for all $b \in (-\infty,0]$ with $\mathbf{v}^-$). The other direction is trivial.

\ref{lemma:char_surj_2})
            We start with some preliminary observations.
            First, there is a constant $C \in \R$ that only depends on the weights and biases of $f$ such that $\|f-f_0\|_\infty \leq C$ (\cite{HBDS21} Proposition~2.3).
		Moreover, due to continuity of~$f$ it holds that $f$ is surjective if and only if for every $a<b\in \R$ there are $a',b' \in \R$ with $a'<a$ and $b' > b$ such that $a',b' \in f(\R^d)$.
		
		Now, for the first direction, assume that $f$ is surjective.
                Then there are $a < -C$ and $b > C$ such that $a,b \in f(\R^d)$.
                Since $\|f-f_0\|_\infty \leq C$, it follows that there are $\mathbf{v}^+$ and $\mathbf{v}^-$ with $f_0(\mathbf{v^+}) > 0$ and $f_0(\mathbf{v}^-) < 0$, implying with~\ref{lemma:char_surj_1}) that $f_0$ is surjective. 
		
		For the converse direction, let $a < 0 < b$. By surjectivity of $f_0$, we have that $a- 2C\in f_0(\R^d)$ and $b + 2C\in f_0(\R^d)$.
                Hence, since $\|f-f_0\|_\infty \leq C$, it follows that there are $a'<a$ and $b'>b$ such that $a',b'\in f(\R^d)$, implying surjectivity of $f$.

\ref{lemma:char_surj_3}) Follows directly from \ref{lemma:char_surj_1}) and \Cref{observation:linearoncells}.
\end{proofwithcaption}
 
 \begin{proofwithcaption}{Proof of Lemma~\ref{lemma:polytime_zeromap}}
 \label{app:lemma:polytime_zeromap}
      	We define the sets $I^+ \coloneqq \{i \in [n] \mid (\mathbf{W}_2)_i = 1\}$ and $I^- \coloneqq [n] \setminus I^+$ and let $\mathbf{w}_1,\ldots,\mathbf{w}_{n}$ be the rows of~$\mathbf{W}_1$.
      	
        First, we can assume that for any $\mathbf{v} \in \R^d$ there is at most one $i \in [n]$ such that $\mathbf{w}_i \in \pos(\mathbf{v}) \coloneqq \{\lambda \mathbf{v} \mid \lambda \geq 0\}$. To see this, let $\mathbf{w}_i \in \pos(\mathbf{w}_j)$ for some $i,j \in [n]$.
        If $i,j \in I^+$, then we can simply delete the rows $\mathbf{w}_i$ and $\mathbf{w}_j$ and add a new row $\mathbf{w}_i + \mathbf{w}_j$ without changing the map~$f$ (clearly the same works for $i,j \in I^-$). If $i \in I^+$ and $j \in I^-$, then we can delete the rows $\mathbf{w}_i$ and $\mathbf{w}_j$ and, if $\|\mathbf{w}_j\|_2 \leq \|\mathbf{w}_i\|_2$, add a new row $\mathbf{w}_i - \mathbf{w}_j$ with output weight $1$ or add a new row $\mathbf{w}_j - \mathbf{w}_i$ with output weight $-1$ if $\|\mathbf{w}_i\|_2 \leq \|\mathbf{w}_j\|_2$ without changing the map $f$.
        Note that we can transform any matrix $\mathbf{W}_1$ to such a form in polynomial time.

        Now, if for every row $\mathbf{w}_j$ there is a row $\mathbf{w}_{j'}$ such that $\mathbf{w}_j=-\mathbf{w}_{j'}$ and $(\mathbf{W}_2)_j = -(\mathbf{W}_2)_{j'}$, it follows that $f$ is a linear map and hence we can easily check whether it is the zero map.
        Otherwise, assume that for a row $\mathbf{w}_j$ there is no such row $\mathbf{w}_{j'}$. Then $\mathbf{w}_{j}$ induces a hyperplane $H_j \coloneqq \{\mathbf{x} \in \R^d \mid \inner{\mathbf{w}_j,\mathbf{x}}=0\}$ such that $\phi_{\mathbf W_1}$ is not linear in every open neighborhood of any $\mathbf{x} \in H_j$ and hence the map $f$  is not linear and in particular cannot be the zero map. Thus, we can check in polynomial time whether $f = 0$.
        
    	Now in the case of $f \neq 0$, let $I_j \coloneqq \{i \in [n] \mid \mathbf{w}_i \in \spn(\mathbf{w}_j)\}$ and note that $|I_j| \leq 2$.	
        For $i \in [n]$, we define the hyperplane  $H_i  \coloneqq  \{\mathbf{x} \in \R^d \mid \inner{\mathbf{w}_i,\mathbf{x}} = 0\}$. By definition, there exists an $\mathbf{x} \in H_j \setminus \left(\bigcup_{i \in [n] \setminus  I_j} H_i\right)$.
        Now, for \[\varepsilon \coloneqq \min \{1,\nicefrac{1}{2}\min\limits_{i \in [n] \setminus I_j} \min\limits_{\mathbf y \in H_i} \|\mathbf{x} - \mathbf{y}\|_2 > 0\}\]
        it holds that $\{\mathbf{x} + \delta \mathbf{w}_j, \mathbf{x} - \delta \mathbf{w}_j\} \subset \R^d  \setminus \left(\bigcup_{i \in [n] \setminus I_j} H_i\right)$ for all $\delta \in (0,\varepsilon)$.
        Let $\mathbf{x}' \coloneqq \mathbf{x} +\varepsilon \mathbf{w}_j$ and $\mathbf{x}'' \coloneqq \mathbf{x} -\varepsilon \mathbf{w}_j$ and let $I' \coloneqq  \{i \in [n] \mid \inner{\mathbf{w}_i,\mathbf{x}}' > 0\}$ and $I'' \coloneqq  \{i \in [n] \mid \inner{\mathbf{w}_i,\mathbf{x}}'' > 0\}$.
        
        We will argue now that either in the cell $C' \in \Sigma_{\mathbf W_1}$ containing $\mathbf{x}'$ or in the cell $C'' \in \Sigma_{\mathbf W_1}$ containing~$\mathbf{x}''$ we find the desired $\mathbf{x}^*$ with $f(\mathbf{x}^*) \neq 0$.
        Note that it is sufficient to prove that~$f$ cannot be the zero map on $C'$ and $C''$.
        We prove this by showing that \[\mathbf{W}_2 \circ (\mathbf{W}_1)_{C'} - \mathbf{W}_2 \circ (\mathbf{W}_1)_{C''} \neq \mathbf{0}.\]
        Note that $(I' \cup I_j) \setminus \{j\} = I''$.
        If $I_j = \{j\}$, then \[\left(\sum_{i \in I' \cap I^+} \mathbf{w}_i -  \sum_{i \in I' \cap I^-} \mathbf{w}_i\right) - \left(\sum_{i \in I'' \cap I^+} \mathbf{w}_i -  \sum_{i \in I'' \cap I^-} \mathbf{w}_i\right) = \pm \mathbf{w}_j\neq  \mathbf 0.\]
        If $I_j = \{j,j'\}$, then we have $\mathbf{w}_{j'} = - \lambda \mathbf{w}_j$ for some $\lambda >0$ and hence \[\left(\sum_{i \in I' \cap I^+} \mathbf{w}_i -  \sum_{i \in I' \cap I^-} \mathbf{w}_i\right) - \left(\sum_{i \in I'' \cap I^+} \mathbf{w}_i -  \sum_{i \in I'' \cap I^-} \mathbf{w}_i\right) = \pm \mathbf{w}_j \pm \mathbf{w}_{j'}\] equals $\mathbf 0$ if and only if $\lambda=1$ and $(\mathbf{W}_2)_j = -(\mathbf{W}_2)_{j'}$, which we assumed not to be the case.
        \end{proofwithcaption}

 \paragraph{Proposition~\ref{prop:poscutNPhard}}
 \label{proof:poscut}
 In order to prove that \textsc{Positive Cut} is \NP-complete, we reduce from the following NP-hard problem (\cite{Bonsma10}).
 	\problemdef{Densest Cut}
	{A graph $G=(V,E)$ and $t \in \Q \cap [0,1]$.}
	{Is there a subset $S \subseteq V$ such that $\frac{|E(S,V\setminus S)|}{|S|\cdot |V\setminus S|}>t$?}
    \begin{proposition}\label{prop:poscutNPhard}
		\textsc{Positive Cut} is \NP-complete.
	\end{proposition}
 \begin{proof} 
		We reduce from \textsc{Densest Cut}. Let $(V,E)$ be a graph and $t= \frac{a}{b}$ with $a < b \in \N$ and let $w(F) \coloneqq \sum_{e \in F} w(e)$.
		We construct the complete graph $K_{|V|} = (V,E'\coloneqq {V\choose 2})$ with edge weights $w \colon E' \to \Z$ by
		\[w(\{i,j\}) = 
		\begin{cases}
			-ab & \{i,j\} \notin E \\
			(b-a)b & \{i,j\} \in E
		\end{cases}.\]                    
		Note that for any $S \subseteq V$, it holds that 
		\begin{align*}
			w(E'(S,V\setminus S)) &=
			(b-a)b
			\cdot |E(S,V\setminus S)| - ab  (|S|\cdot|V\setminus S| - |E(S,V\setminus S)|)  \\&= b^2 \cdot |E(S,V\setminus S)| - ab \cdot |S| \cdot |V\setminus S|. 
		\end{align*}
		Hence, we have
		\[w(E'(S,V \setminus S)) > 0 \iff \frac{b^2 \cdot |E(S,V\setminus S)|}{ab \cdot |S|\cdot |V\setminus S| } > 1 \iff \frac{ |E(S,V\setminus S)|}{|S| \cdot |V\setminus S| } > \frac{a}{b}=t,\]
		proving the correctness of the reduction.
	\end{proof}
        
\begin{proofwithcaption}{Proof of Theorem~\ref{thm:Posi-NP}}
\label{app:thm:Posi-NP}
          The problem is contained in NP since, by \Cref{lemma:charsurjectivity},
          it is sufficient to define a ray $\rho$ as a certificate for positivity.
          Since rays are one-dimensional subspaces, they are the intersection of $d-1$ hyperplanes corresponding to rows of $\mathbf{W}_1$.
		Hence, the~$d-1$ rows of~$\mathbf W_1$ that determine~$\rho$ form a polynomial-time verifiable certificate.
		
		For the \NP-hardness, we reduce from \textsc{Positive Cut}. 
		Given a weighted graph $(G=(V,E),w)$ with $V = [d]$ and $|E|=n$, we define the matrices $\mathbf W_1\in\R^{2n\times d}$ and $\mathbf W_2\in\R^{1\times 2n}$ as follows:
                For each~$e=\{i,j\}\in E$, $\mathbf W_1$ contains two rows $\mathbf w_e\coloneqq \mathbf e_i -\mathbf e_j$ and $\mathbf w_e'\coloneqq -\mathbf w_e$.
              The corresponding entries of~$\mathbf W_2$ are set to $w(e)$.
		Thus, the 2-layer ReLU neural network computes the map $f\colon\R^d\to\R$ with
		\[f(\mathbf x)=\sum_{\{i,j\}\in E} w(\{i,j\}) \cdot ([x_i-x_j]_+ + [x_j-x_i]_+).\]
		
		For the correctness, we start with some preliminary observations.
		For a subset $S \subseteq V$, let $\mathbf r_S \coloneqq  \sum_{i \in S} \mathbf e_i \in\R^d$ and $\mathbf{r}_S' \coloneqq - \sum_{i \in V\setminus S} \mathbf e_i \in\R^d$
		and note that
		\[ [(\mathbf r_S)_i-(\mathbf r_S)_j]_+ + [(\mathbf r_S)_j-(\mathbf r_S)_i]_+ = 
		\begin{cases}
			0, & \{i,j\} \notin E(S,V \setminus S) \\
			1, & \{i,j\} \in E(S,V \setminus S)
		\end{cases}\] 
		as well as 
		\[ [(\mathbf{r}_S')_i-(\mathbf{r}_S')_j]_+ + [(\mathbf{r}_S')_j-(\mathbf{r}_S')_i]_+ = 
	\begin{cases}
		0, & \{i,j\} \notin E(S,V \setminus S) \\
		1, & \{i,j\} \in E(S,V \setminus S)
	\end{cases}\]  
		and therefore
		\[f(\mathbf r_S)=f(\mathbf{r}_S')=\sum\limits_{\{i,j\}\in E(S,V \setminus S)} w(\{i,j\})=w(E(S,V \setminus S)).\]
		
		As regards the correctness, if~$(G,w)$ is a yes-instance, that is, there exists a subset~$S\subseteq V$ with~$w(E(S,V\setminus S))>0$, then~$f(\mathbf r_S)>0$.
		
		Conversely, assume that there is a $\mathbf{v} \in \R^d$ with~$f(\mathbf v) > 0$.
		Let~$\pi\in\mathcal S_{d+1}$ be a permutation such that~$v_{\pi(1)}\le v_{\pi(2)}\le \cdots\le v_{\pi(d+1)}$, where $v_{d+1} \coloneqq 0$.
		Since all hyperplanes defined by hidden neurons are of the form $\{\mathbf x\in\R^d\mid x_i=x_j\}$ for some $i$ and $j$, the map $\phi_{\mathbf W_1}$ is linear within the pointed $d$-dimensional cone \[C \coloneqq \{\mathbf{x} \in \R^d \mid x_{\pi(1)} \leq \cdots \leq x_{\pi(d+1)}\} = \bigcap_{i=1}^{d} \{\mathbf{x} \in \R^d \mid x_{\pi(i)} \leq x_{\pi(i+1)}\},\] where again $x_{d+1} \coloneqq 0$. 
		Hence, by linearity of the output layer and \Cref{observation:linearoncells}, the value~$f(\mathbf v)$ is a conical combination of the values of $f$ on the ray generators of the extreme rays of~$C$ and hence there is a ray generator $\mathbf{r}$ in $C$ such that $f(\mathbf{r}) > 0$.
                All extreme rays of $C$ are the intersection of $C$ and $d-1$ hyperplanes of the form $\{\mathbf{x} \in \R^d \mid x_{\pi(i)} = x_{\pi(i+1)}\}$.
                We denote these rays by \[\rho_k \coloneqq \{\mathbf x \in \R^d \mid  x_{\pi(1)} = \cdots = x_{\pi(k-1 )}\leq x_{\pi(k)} = \cdots = x_{\pi(d+1)} \}\]
                for $k\in[d]$.
		Let $k\in[d]$ and $S_k\coloneqq \{\pi(k), \ldots, \pi(d+1)\}$.
                If $d+1\not\in S_k$, then $\mathbf{r}_{S_k}$ generates $\rho_k$ and otherwise $\mathbf{r}_{S_k}'$ generates $\rho_k$.
                Since $f$ is positive on one of these $d$ ray generators, we can conclude that there is an $S \subseteq V$ such that $f(\mathbf{r}_S) = f(\mathbf{r}_S') > 0$ which implies~$w(E(S,V\setminus S))>0$.
	\end{proofwithcaption}

\begin{proofwithcaption}{Proof of \Cref{prop:zonotope_containment}}
\label{app:zonotope_containment}
Given the map $f(\mathbf{x})=  \mathbf{W}_2 \cdot \phi_{\mathbf{W}_1}(\mathbf{x})$ (where~$\mathbf W_2\in\{-1,1\}^{1\times n}$), we define the sets $I^+ \coloneqq \{i \in [n] \mid (\mathbf{W}_2)_i = 1\}$ and $I^- \coloneqq [n] \setminus I^+$ and let $\mathbf{W} \coloneqq \mathbf{W}_1$.
       Note that we have 
	   \[f(\mathbf{x}) = \sum_{i \in I^+} \max \{0,\inner{\mathbf{w}_i,\mathbf{x}}\} -\sum_{i \in I^-} \max \{0,\inner{\mathbf{w}_i,\mathbf{x}}\},\]
	   and therefore $f = \varphi^{-1}(Z^+) - \varphi^{-1}(Z^-)$, where 
	   \[Z^+ \coloneqq \sum_{i \in I^+} \conv \{\mathbf 0,\mathbf{w}_i\}= Z(\mathbf{W}_{I^+}),\] and  \[Z^- \coloneqq \sum_{i \in I^-} \conv \{\mathbf 0,\mathbf{w}_i\}= Z(\mathbf{W}_{I^-}).\]
       Note that the support functions $\varphi^{-1}(Z^+)$ and  $\varphi^{-1}(Z^-)$ can only attain nonnegative values and $Z^+\subseteq Z^-$ implies that $\varphi^{-1}(Z^+) \leq \varphi^{-1}(Z^-)$. Moreover, if for a $\mathbf{v} \in Z^+$, it holds that $\varphi^{-1}(Z^+)(\mathbf{v}) > \varphi^{-1}(Z^-)(\mathbf{v})$, then $\mathbf{v} \notin Z^-$. Therefore, there exists a $\mathbf{v} \in \R^d$ such that $\varphi^{-1}(Z^+)(\mathbf{v}) > \varphi^{-1}(Z^-)(\mathbf{v})$ if and only if $Z^+ \nsubseteq Z^-$. Since any two zonotopes are of this form, this concludes the proof.
       \end{proofwithcaption}

\end{document}